\documentclass[10pt,twocolumn]{IEEEtran}
\usepackage{amsmath,amssymb,latexsym,cite}
\usepackage{pstricks,graphicx}
\usepackage{pst-node,pst-plot,pstricks-add}

\usepackage{pstricks,graphicx,subfigure}
\usepackage{latexsym,cite}

\newtheorem{theorem}{Theorem}

\newtheorem{lemma}{Lemma}

\newtheorem{example}{Example}

\def\prob{\mathbb{P}}
\def\expe{\mathbb{E}}
\renewcommand{\bar}{\overline}

\newcommand{\kldiv}[2]{\ensuremath{D\left(#1\ \left\|\ #2 \right. \right)}}

\newcommand{\empMI}[2]{\ensuremath{I\left(#1, #2 \right)}}
\newcommand{\maxvar}[2]{\ensuremath{d_{\max}\left(#1, #2 \right)}}

\def\mbf{\mathbf}
\def\mc{\mathcal}

\def\scostf{l}
\def\scostb{\Lambda}
\def\scostmax{\lambda^{\ast}}

\def\denc{\phi}
\def\ddec{\psi}
\def\renc{\Phi}
\def\rdec{\Psi}

\def\decreg{D}
\def\rdecreg{\mathbf{D}}
\def\stdmaxerr{\varepsilon}
\def\nosymaxerr{\hat{\varepsilon}}

\def\listmaxerr{\varepsilon_L}

\def\nosyjam{J}
\newcommand{\chunkerr}[1]{\delta_{#1}}

\newcommand{\maxerrs}[2]{\ensuremath{\stdmaxerr(#1,#2)}}

\def\robrlerr{\nosymaxerr}

\newcommand{\typ}[1]{\ensuremath{T_{#1}}}
\newcommand{\etyp}[2]{\ensuremath{T_{#1}^{#2}}}
\newcommand{\btyp}[2]{\ensuremath{\mc{T}_{#2}({#1})}}
\newcommand{\shell}[3]{\ensuremath{T_{#1}^{#2}(#3)}}

\def\wstd{\mc{W}_{\mathrm{std}}}
\def\wdep{\mc{W}_{\mathrm{dep}}}

\def\mlconeps{\epsilon_{4}}

\def\nosyeps{\epsilon}
\def\listloss{\epsilon_{1}}
\def\autherr{\epsilon}

\def\stdrleps{\epsilon}
\def\stdrlloss{\delta}

\def\robrlloss{\epsilon}
\def\robrlestgap{\delta}
\def\robrlgap{\xi}

\def\chunk{\ensuremath{c}}

\newcommand{\cvar}[2]{\ensuremath{\mathbf{#1}^{(#2 \chunk)}}}
\newcommand{\csi}[1]{\mc{V}_{#1}}
\newcommand{\normcsi}[1]{\hat{\lambda}_{#1}}
\newcommand{\truenorm}[1]{\lambda_{#1}}
\newcommand{\csitot}[1]{\hat{\Lambda}_{#1}}
\newcommand{\truetot}[1]{\Lambda_{#1}}
\newcommand{\codenorm}[1]{\tilde{\Lambda}_{#1}}

\newcommand{\posscsi}[1]{\mathbb{V}(#1)}

\def\csiexp{v}
\def\dectime{\mbf{M}}
\def\possdec{\mc{M}}
\def\Mhi{M^{\ast}}
\def\Mlo{M_{\ast}}

\def\minrate{R_{\min}}
\def\decide{\tau}
\def\remp{R_{\mathrm{emp}}}

\def\capstd{\ensuremath{C_{\mathrm{std}}}}
\def\capdep{\ensuremath{C_{\mathrm{dep}}}}

\def\hb{\ensuremath{h_b}}

\newrgbcolor{dkgray}{0.3 0.3 0.3}

\title{Rateless codes for AVC models}
\author{
Anand~D.~Sarwate~\IEEEmembership{Member,~IEEE,} and~Michael~Gastpar~\IEEEmembership{Member,~IEEE}
\thanks{Manuscript received December 7, 2007;
accepted September 30, 2009.  }
\thanks{Part of this work was presented at the 2007 IEEE Symposium on Information Theory in Nice, France \cite{SarwateG:07nosy}, and at the 2007 IEEE Information Theory Workshop in Tahoe, CA \cite{SarwateG:07rateless}.  }
\thanks{A.D. Sarwate was with the Department of Electrical Engineering and Computer Sciences, University of California, Berkeley, and is now with the Information Theory and Applications Center, University of California, San Diego, La Jolla, CA 92093-0447 USA.}
\thanks{M. Gastpar is with the Department of Electrical Engineering and Computer Sciences, University of California, Berkeley, Berkeley CA 94720-1770 USA.}
\thanks{The work of A.D. Sarwate and M. Gastpar was supported by the National Science Foundation under award CCF-0347298.}
}
\date{\today}

\begin{document}

\maketitle

\begin{abstract}
The arbitrarily varying channel (AVC) is a channel model whose state
is selected maliciously by an adversary.  Fixed-blocklength coding assumes
a worst-case bound on the adversary's capabilities, which leads to pessimistic results.  This paper defines a variable-length perspective on this problem, for which achievable rates are shown that depend on the realized actions of the adversary.
Specifically, rateless codes are constructed which require a limited amount of common randomness.  These codes are constructed for two kinds of AVC models.  In the first the channel state cannot depend on the channel input, and in the second it can.  As a byproduct, the randomized coding capacity of the AVC with state depending on the transmitted codeword is found and shown to be achievable with a small amount of common randomness.  The results for this model are proved using a randomized strategy based on list decoding.
\end{abstract}

\section{Introduction}

Modern communication platforms such as sensor networks, wireless ad-hoc networks, and cognitive radio involve communication in environments that are difficult to model.  This difficulty may stem from the cost of measuring channel characteristics, the behavior of other users, or the interaction of heterogeneous systems using the same resources.  These complex systems may use extra resources such as feedback on a low-rate control channel or common randomness to overcome this channel uncertainty.  We are interested in how such resources can be used to deal with interference that is difficult to model or which may depend on the transmitted codeword.

Inspired by some of these challenges, we approach the problem from the perspective of variable-length coding over arbitrarily varying channels (AVCs).  The AVC is an adversarial channel model in which the channel is governed by a time varying state controlled by a \textit{jammer} who wishes to maximize the decoding error probability.  For fixed-blocklength coding, the capacity is the worst-case over all allowable actions of the jammer.  However, in some cases the worst-case may be unduly pessimistic.  Correspondingly, we ask the following questions : can variable-length codes be developed for AVC models that adapt to the realized actions of the jammer?  How much feedback and common randomness is needed to enable these codes?

In this paper we study \textit{randomized coding} for two different models based on the AVC.  In a randomized code the encoder and decoder have a shared source of common randomness unknown to the jammer.  This common randomness acts as a shared \textit{key} to mask the coding strategy from the jammer.  The first model we study is the AVC under maximal error and randomized coding, in which the state sequence is chosen independently of the transmitted codeword.  The second model is an AVC in which the jammer can choose the state sequence based on the transmitted codeword.  This may be an appropriate model for a multi-hop network in which an internal node becomes compromised and tampers with transmitted packets.  We call this situation an AVC with ``nosy noise.'' Our first result is a formula for the randomized coding capacity of this AVC.  Our proof uses results on list decoding for AVCs \cite{Ahlswede:73list,Ahlswede:93list,Sarwate:08thesis} with a partial derandomization technique used by Langberg \cite{Langberg:04focs}.

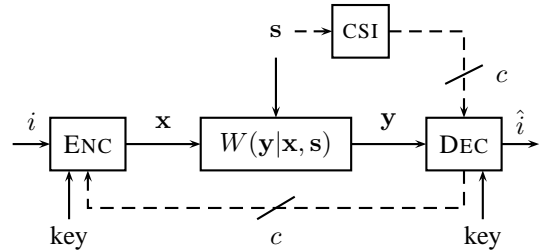
\begin{figure}
\begin{center}
\psset{xunit=0.05cm,yunit=0.05cm,runit=0.05cm}
\begin{pspicture}(0,5)(140,80)

\newrgbcolor{pink}{1.0 0.5 0.5}
\newrgbcolor{babyblue}{0.8 0.8 1.0}
\newrgbcolor{lavender}{1.0 0.7 1.0}

\psline[arrows=->](0,40)(10,40)
\psframe(10,33)(30,47)
\rput(20,40){\textsc{Enc}}

\psline[arrows=->](30,40)(50,40)
\psframe(50,33)(90,47)
\rput(70,40){$W(\mbf{y} | \mbf{x},\mbf{s})$}

\psline[arrows=->](90,40)(110,40)
\psframe(110,33)(130,47)
\rput(120,40){\textsc{Dec}}

\psline[arrows=->](130,40)(140,40)

\rput(5,46){$i$}
\rput(40,46){$\mbf{x}$}
\rput(100,46){$\mbf{y}$}
\rput(135,46){$\hat{i}$}

\psline[arrows=->](70,63)(70,47)
\rput(70,70){$\mbf{s}$}
\psline[arrows=->,linestyle=dashed](75,70)(85,70)
\psframe(85,63)(100,77)
\rput(92.5,70){\textsc{csi}}
\psline[arrows=->,linestyle=dashed](100,70)(120,70)(120,47)
\psline(115,56)(125,61)
\rput(130,58){$\chunk$}

\psline[arrows=<-,linestyle=dashed](20,33)(20,23)(120,23)(120,33)
\psline(65,20)(75,26)
\rput(70,15){$\chunk$}

\psline[arrows=->](15,20)(15,33)
\psline[arrows=->](125,20)(125,33)
\rput(15,15){key}
\rput(125,15){key}

\end{pspicture}
\caption{Rateless communication system.  The encoder and decoder share a source of common randomness.  A single bit of feedback is available every $\chunk$ channel uses for the decoder to terminate transmission.  Some partial information about the channel state is available at the decoder every $\chunk$ channel uses in a causal fashion. \label{fig:rateless}}
\end{center}
\end{figure}
The main focus of this paper is on the problem of \textit{rateless coding} for these channels using \textit{limited common randomness} and \textit{partial channel state information}, as shown in Figure \ref{fig:rateless}.  Rateless codes were first proposed for erasure channels \cite{Luby:02LT,Shokrollahi:03fountain} and compound channels \cite{Shulman:03commonbc,TchamkertenT:06variable}, and a general model is discussed in \cite{ShamaiTV:07fountain}.  They are strategies that allow a single-bit feedback signal (often called an ACK/NACK for ``acknowledge''/``not acknowledge'') every $\chunk$ channel uses to terminate transmission based on the observed channel output $\mbf{y}$ and channel state information.  In our model, the partial state information takes the form of estimates of the average channel induced by the channel state $\mbf{s}$ over ``chunks'' of size $\chunk$.  In practice this channel information may come from exogenous measurements or from training information in the forward link, as in \cite{EswaranSSG:07indseq}.

We propose a model for partial state information at the decoder which consists of an estimate of the empirical channel.  We then provide partially derandomized rateless code constructions for the two AVC models.  These codes have fixed input type and are piecewise constant-composition, and for accurate partial state information can achieve rates close to the mutual information of a corresponding AVC.  The derandomization for these codes comes from strategies of Ahlswede \cite{Ahlswede:78elimination} and Langberg \cite{Langberg:04focs}.

\subsection*{Related work and context}

The arbitrarily varying channel was first studied in the seminal paper of Blackwell, Breiman, and Thomasian \cite{BlackwellBT:60random}, who found a formula for the capacity under randomized coding and maximal error.  Without randomized coding, the maximal-error problem is significantly harder \cite{KieferW:62avc,AhlswedeW70:binary,Ahlswede80:method,CsiszarK:81max} and is related to the the zero-error capacity \cite{Ahlswede:70weak}.  The AVC model was extended to include constraints on the jammer by Hughes and Narayan \cite{HughesN:87gavc} and Csisz\'{a}r and Narayan \cite{CsiszarN:88constraints,CsiszarN:88positivity,CsiszarN:91gavc}.  For randomized coding, error exponents have also been studied \cite{Ericson:85exponent,HughesT:96exponent,ThomasH:91exponent}.
 
Ahlswede's landmark paper \cite{Ahlswede:78elimination} showed that the average error capacity under deterministic coding $\bar{C}_d$ is $0$ or equal to the randomized coding capacity $C_r$.  Randomized coding gives the same capacity under maximal and average error, but for deterministic coding under average error the capacity may be positive and strictly smaller than the randomized coding capacity when cost constraints are involved \cite{CsiszarN:88positivity}.  However, Ahlswede's technique can be used to show that only $O(\log n)$ bits of common randomness is needed to achieve $C_r(\scostb)$ for AVCs with cost constraint $\scostb$.

In the ``nosy noise'' model, shown in Figure \ref{fig:nosyavc}, has been discussed previously in the AVC literature, where it is sometimes called the A$^{\ast}$VC.  For deterministic coding, knowing the message is the same as knowing the codeword, so the maximal error capacity is as the nosy noise capacity \cite[Problem 2.6.21]{CsiszarKorner}.  In some cases the average error capacity is also the same \cite{AhlswedeW70:binary}.  The capacity under noiseless feedback was later found by Ahlswede \cite{Ahlswede:73fback}.  To our knowledge, for cost-constrained AVCs the problem was not studied until Langberg \cite{Langberg:04focs} found the capacity for bit-flipping channels with randomized coding.  Smith \cite{Smith:07scrambling} has shown a computationally efficient construction using $O(n)$ bits of common randomness.  Agarwal, Sahai and Mitter proposed a similar model with a distortion constraint \cite{AgarwalSM:06allerton}, which is different than the AVC model considered here \cite{Sarwate:08thesis}.

Our study of rateless codes is inspired by hybrid-ARQ \cite{Soljanin:03harq} and recent work that has shown how zero-rate feedback can improve channel reliability \cite{Sahai:08delay,DraperAllerton06,Sahai:07balancing}.  In \cite{EswaranSSG:07indseq} the encoder and decoder use randomly placed training sequences to estimate the channel quality.  Another inspiration was the paper of Draper et. al \cite{DraperFK:05rateless}, which studies an AVC model where the entire state sequence given to the decoder as side information and single-bit feedback acts as an ACK/NACK to terminate decoding.  Our coding schemes can be used to provide a component of the coding strategy of \cite{EswaranSSG:07indseq}, which shows that the rates achievable by Shayevitz and Feder \cite{ShayevitzF:06empirical} for individual sequence channels are also achievable with zero-rate feedback.  

In the next section we describe the channel model and in Section \ref{sec:main} we state the main contributions of this paper.  The two derandomization strategies are discussed in Section \ref{sec:derand}, where we also find the capacity of AVCs with ``nosy noise.''  Sections \ref{sec:std_rateless} and \ref{sec:rob_rateless} contain our rateless code constructions for channels with input-independent and input-dependent state, respectively.

\section{Channel models and definitions \label{sec:defs}}
We will model our time-varying channel by a set of channels $\mc{W} = \{W(y | x, s) : s \in \mc{S}\}$ with finite input alphabet $\mc{X}$, output alphabet $\mc{Y}$, and constrained state sequence \cite{CsiszarN:88constraints}.  This is an arbitrarily varying channel (AVC) model.  If $\mbf{x} = (x_1, x_2, \ldots, x_n)$, $\mbf{y} = (y_1, y_2, \ldots, y_n)$ and $\mbf{s} = (s_1, s_2, \ldots, s_n)$ are length $n$ vectors, the probability of observing the output $\mbf{y}$ given the input $\mbf{x}$ and state $\mbf{s}$ over the AVC $\mc{W}$ without feedback is given by:
	\begin{align}
	W(\mbf{y}| \mbf{x}, \mbf{s}) = \prod_{i=1}^{n} W(y_i | x_i, s_i)~.
	\label{eq:blockchannel}
	\end{align}
In this paper the feedback is used only to terminate transmission, and we compare our achievable rates with those achievable without feedback (c.f. \cite{ShamaiTV:07fountain}).  The interpretation of (\ref{eq:blockchannel}) is that the channel state can change arbitrarily from time to time.  The AVC is an adversarial model in which the state is controlled by a \textit{jammer} who wishes to stymie the communication between the encoder and decoder.   As we will see, the knowledge held by the adversary can be captured in the error criterion.

One extension of this model is to introduce constraints on the input and state sequences \cite{CsiszarN:88constraints}.  For simplicity we will only assume constraints on the state.  Let $\scostf : \mc{S} \to \mathbb{R}^{+}$ be a cost function on the state set, where $\min_{s} \scostf(s) = 0$ and $\max_{s \in \mc{S}} \scostf(s) = \scostmax < \infty$.  The cost of the vector $\mbf{s} = (s_1, s_2, \ldots, s_n)$ is the sum of the cost on the elements:
	\begin{align}
	\scostf(\mbf{s}) &= \sum_{i=1}^{n} \scostf(s_i)~.
	\end{align}
In some cases we will impose a total constraint $\scostb$ on the average cost, so that
	\begin{align}
	\scostf(\mbf{s}) &\le n \scostb~. 
	\label{eq:costconstraint}
	\end{align}
If $\scostb \ge \scostmax$ we say the state is unconstrained.  We will define the set 
	\begin{align}
	\mc{S}^n(\scostb) = \{\mbf{s} : \scostf(\mbf{s}) \le n \scostb\}
	\label{eq:validS}
	\end{align} 
to be the set of sequences with average cost less than or equal to $\scostb$.

\subsection{Point-to-point channel coding}

A $(n,N)$ \textit{deterministic code} $\mc{C}$ for the AVC $\mc{W}$ is a pair of maps $(\denc, \ddec)$ with $\denc : [N] \to \mc{X}^n$ and $\ddec : \mc{Y}^n \to [N]$.  The \textit{rate} of the code is $n^{-1} \log N$.  The \textit{decoding region} for message $i$ is $\decreg_{i} = \{\mbf{y} :  \ddec(\mbf{y}) = i\}$ \label{def:decreg}.  We can also write a deterministic code $\mc{C}$ as a set of pairs $\{ (\mbf{x}(i), \decreg_{i}) : i \in [N]\}$ with the encoder $\denc$ and decoder $\ddec$ defined implicitly.  The \textit{error for message $i$ and state sequence $\mbf{s} \in \mc{S}^n(\scostb)$} is given by
	\begin{align}
	\stdmaxerr(i,\mbf{s}) = 1 - W\left( \decreg_{i} | \mbf{x}(i), \mbf{s} \right)~.
	\label{not:fixedSerr}
	\end{align}

\begin{figure}
\begin{center}
\psset{xunit=0.05cm,yunit=0.05cm,runit=0.05cm}
\begin{pspicture}(-10,5)(140,80)

\psline[arrows=->](-10,40)(0,40)
\psframe(0,33)(20,47)
\rput(10,40){\textsc{Enc}}

\psline[arrows=->](20,40)(50,40)
\psframe(50,33)(90,47)
\rput(70,40){$W(\mbf{y} | \mbf{x},\mbf{s})$}

\psline[arrows=->](90,40)(110,40)

\psframe(110,33)(130,47)
\rput(120,40){\textsc{Dec}}
\psline[arrows=->](130,40)(140,40)

\rput(-5,46){$i$}
\rput(35,46){$\mbf{x}(i,k)$}
\rput(100,46){$\mbf{y}$}
\rput(135,46){$\hat{i}$}

\psframe(60,63)(80,77)
\rput(70,70){\textsc{Jam}}
\psline[arrows=->](70,63)(70,47)
\rput(75,55){$\mbf{s}$}

\psline[arrows=<->,linestyle=dashed](10,33)(10,23)(120,23)(120,33)
\psline[linestyle=dashed](70,13)(70,23)
\rput(70,7){$k \in \{1, 2, \ldots, K\}$}

\end{pspicture}
\caption{An arbitrarily varying channel with randomized encoding.   The encoder and decoder share a secret key in $[K]$ that is unknown to the jammer.  \label{fig:randavc}}
\end{center}
\end{figure}
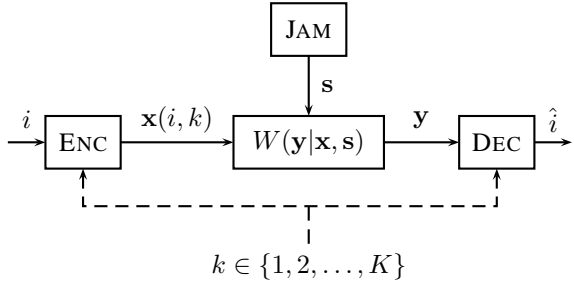
A $(n,N)$ \textit{randomized code} $\mbf{C}$ \label{not:randcode} for the AVC $\mc{W}$ is random variable taking on values in the set of deterministic codes.  It is written as a pair of random maps $(\renc,\rdec)$ where each realization is an $(n,N)$ deterministic code.  If $(\renc, \rdec)$ almost surely takes values in a set of $K$ codes, then we call this an $(n,N,K)$ randomized code.  We can also think of an $(n,N,K)$ randomized code as a family of codes $\{(\denc_k, \ddec_k) : k \in [K] \}$ indexed by a set of $K$ \textit{keys}, as shown in Figure \ref{fig:randavc}.  The \textit{key size} of a randomized code $(\renc,\rdec)$ is the entropy $H(\mbf{C})$ of the code.  In the case where $\mbf{C}$ is uniformly distributed on a set of $K$ codes, the key size is simply $\log K$.   Note that the realization of the code is shared by the encoder and decoder, so the key is known by both parties.  The \textit{rate} of the code is $R = n^{-1} \log N$.  The \textit{decoding region} for message $i$ under key $k$ is $D_{i,k} = \{\mbf{y} :  \ddec_k(\mbf{y}) = i\}$.  In the case where the bound on $K$ is not explicit or unspecified, we write the random decoding region for message $i$ as $\rdecreg_i = \{\mbf{y} : \rdec(\mbf{y}) = i\}$.

For a randomized code we require that the decoder error to be small for each message message \textit{averaged over key values}.  Randomization allows several different codewords to represent the same message.  For maximal error, there are two cases to consider, depending on whether or not the state can depend on the actual \textit{codeword}.

\begin{figure}
\begin{center}
\psset{xunit=0.05cm,yunit=0.05cm,runit=0.05cm}
\begin{pspicture}(-10,5)(140,80)

\psline[arrows=->](-10,40)(0,40)
\psframe(0,33)(20,47)
\rput(10,40){\textsc{Enc}}

\psline[arrows=->](20,40)(50,40)
\psframe(50,33)(90,47)
\rput(70,40){$W(\mbf{y} | \mbf{x},\mbf{s})$}

\psline[arrows=->](90,40)(110,40)

\psframe(110,33)(130,47)
\rput(120,40){\textsc{Dec}}
\psline[arrows=->](130,40)(140,40)

\rput(-5,46){$i$}
\rput(35,46){$\mbf{x}(i,k)$}
\rput(100,46){$\mbf{y}$}
\rput(135,46){$\hat{i}$}

\psline[arrows=->](35,53)(35,67)(60,67)
\psframe(60,63)(80,77)
\rput(70,70){\textsc{Jam}}
\psline[arrows=->](70,63)(70,47)
\rput(75,55){$\mbf{s}$}

\psline[arrows=<->,linestyle=dashed](10,33)(10,23)(120,23)(120,33)
\psline[linestyle=dashed](70,13)(70,23)
\rput(70,7){$k \in \{1, 2, \ldots, K\}$}

\end{pspicture}
\caption{The nosy noise error model -- the jammer knows the codeword $\denc_k(i)$.  \label{fig:nosyavc}}
\end{center}
\end{figure}
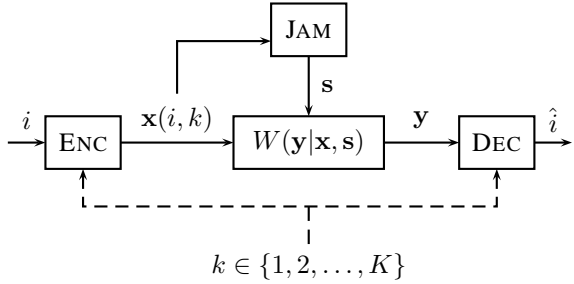

The \textit{standard maximal error} for a $(n,N)$ randomized code over an AVC $\mc{W}$ with cost constraint $\scostb$ is given by
	\begin{align}
	\stdmaxerr &= \max_{i} \max_{\mbf{s} \in \mc{S}^n(\scostb)}
		\expe\left[ 1 
			- W\left( \rdecreg_i | \renc(i), \mbf{s} \right) 
		\right]~,
	\label{eq:maxerrdef}
	\end{align}
where the expectation is over the randomized code $(\renc,\rdec)$.  Here the variables $\rdecreg_i$ and $\renc(i)$ correspond to the same realization of the key.  The \textit{nosy maximal error} for a $(n,N)$ randomized code over an AVC $\mc{W}$ with cost constraint $\scostb$ is given by	
	\begin{align}
	\nosymaxerr &= \max_{i} \max_{\nosyjam : \mc{X}^n \to \mc{S}^n(\scostb)}
		\expe\left[ 1 
			- W\left( \rdecreg_i | \renc(i), \nosyjam(\renc(i)) \right) 
		\right]~,
		\label{eq:nosyerr}
	\end{align}
where the expectation is over the randomized code $(\renc,\rdec)$.  Again, the variables $\rdecreg_i$, $\renc(i)$, and $\nosyjam(\renc(i))$ correspond to the same realization of the key.  We call an AVC under the nosy maximal error criterion an \textit{AVC with nosy noise}.  Figure \ref{fig:nosyavc} shows the channel model under the nosy noise assumption.  In the AVC with nosy noise, the jammer's strategies take the form of mappings $\nosyjam : \mc{X}^n \to \mc{S}^n(\scostb)$ from the codeword vectors to state sequences.  This is a more pessimistic assumption on the jammer's capabilities, since it assumes that it has noncausal access to the transmitted codeword.  Under randomized coding we will show that from a capacity standpoint all that matters is whether the jammer has access to the \textit{current} input symbol.
 
A rate $R$ is called achievable if for every $\epsilon > 0$ there exists a sequence of $(n, N)$ codes of rate $R_n \ge R - \delta$ whose probability of error (maximal or nosy) is at most $\epsilon$.  Whether $R$ is achievable will depend on the error criterion (maximal or nosy).  For a given error criterion, the supremum of achievable rates is the capacity of the arbitrarily varying channel.  We will write $C_r(\scostb)$ for the randomized coding capacity under maximal error with constraint $\scostb$, and $\hat{C}_r(\scostb)$ for the randomized coding capacity with nosy noise and state constraints.

\subsection{Information quantities}

For a fixed input distribution $P(x)$ on $\mc{X}$ and channel $V(y | x)$, we will use the notation $\empMI{P}{V}$ to denote the mutual information between the input and output of the channel.
For a finite or closed and convex set of channels $\mc{V}$ we use the shorthand
	\begin{align}
	\empMI{P}{\mc{V}} = \min_{V \in \mc{V}} \empMI{P}{V}~.
	\end{align}

We define the following sets:
	\begin{align}
		\mc{Q}(\scostb) &= 
			\left\{ Q \in \mc{P}(\mc{S}) : 
				\sum_{s} Q(s) l(s) \le \scostb \right\} 
			\label{eq:def:allowQ} \\
		\mc{U}(P, \scostb) &= 
			\left\{ U \in \mc{P}(\mc{S}|\mc{X}) : 
				\sum_{s,x} U(s|x) P(x) l(s) \le \scostb \right\}~.
			\label{eq:def:allowU}
	\end{align}
For an AVC $\mc{W} = \{W(y|x,s) : s \in \mc{S} \}$ with state constraint $\scostb $ we define two sets of channels: 
	\begin{align}
	\wstd(\scostb) &= 
		\Bigg\{V(y|x) = 
			\sum_{s} W(y|x,s) Q(s) : \nonumber \\
			&\hspace{3cm}
				\ \ Q(s) \in \mc{Q}(\scostb) \Bigg\} 
		\label{eq:indhull} \\
	\wdep(P, \scostb) &= 
		\Bigg\{V(y|x) = 
			\sum_{s} W(y|x,s) U(s|x) : \nonumber \\
			&\hspace{3cm}
				\ \ U(s|x) \in \mc{U}(P, \scostb) \Bigg\}~.
		\label{eq:dephull}
	\end{align}
We will suppress the explicit dependence on $\scostb$.  The set in (\ref{eq:indhull}) is called the \textit{convex closure} of $\mc{W}$, and the set in (\ref{eq:dephull}) is the \textit{row-convex closure} of $\mc{W}$.  In earlier works $\wdep(P, \scostb)$ is sometimes written as $\bar{\bar{\mc{W}}}$ \cite{Ahlswede:78elimination}.

Two information quantities of interest in randomized coding for AVCs are
	\begin{align}
	\capstd(\scostb) &= \max_{P}
		\min_{V \in \wstd(\scostb)} \empMI{P}{V} 
	\label{eq:capstd} \\
	\capdep(\scostb) &= \max_{P}
		\min_{V \in \wdep(P,\scostb)} \empMI{P}{V}~.
	\label{eq:capdep}
	\end{align}
Csisz\'{a}r and Narayan \cite{CsiszarN:88constraints} showed that the randomized coding capacity under maximal error $C_r(\scostb)$ is equal to $\capstd(\scostb)$.  In Theorem \ref{thm:nosynoise} we show that the randomized coding capacity under nosy noise $\hat{C}_r(\scostb)$ is equal to $\capdep(\scostb)$.

\subsection{Rateless codes}

In a rateless code, the decoder can choose to decode at different times based on its observation of the channel output.  We assume that the decoder can inform the encoder that it has decoded in order to terminate transmission.  To simplify the analysis, we consider rateless codes that operate in chunks of length $\chunk(n)$.  For a vector $\mbf{z}$ let $\mbf{z}_{1}^{r}$ denote $(z_1, z_2, \ldots, z_{r})$, and $\cvar{z}{m}$ denote the $m$-th chunk $(z_{(m-1)\chunk+1}, \ldots, z_{m \chunk})$.

The key quantity is the time at which the decoder attempts to decode, which we will
denote by $\dectime\chunk(n),$ {\em i.e.,} decoding is attempted after $\dectime$ chunks.
If this decoding time is appropriately chosen,
then the decoding is successful (with high probability); the corresponding empirical rate is given by
	\begin{align}
	 \remp = \frac{1}{\dectime\chunk} \log_2 N,
	 \label{eq:remp}
	\end{align}
where $N$ is the number of codewords in the codebook.  Defining a rateless code involves not only a codebook, but also a rule according to which the decoder selects the appropriate decoding time $\dectime$.  In our considerations, the decoder performs this selection based on side information about the true channel state (and thus, about the actions of the adversary), which the decoder receives at the end of each chunk.  

More formally, we denote the partial side information (channel estimate) given to the decoder after the $m$-th chunk by $\csi{m}$ \label{def:csi_set}, which takes values in a set $\posscsi{\chunk}$.  We describe the side information model in Section \ref{sec:csi}.  A $(\chunk,N,K)$ randomized rateless code is set of maps $\{(\renc_m, \decide_m, \rdec_m) : m = 1, 2, \ldots\}$:
	\begin{align}
	\renc_m &: [N] \times [K] \to \mc{X}^{\chunk} 
		\label{eq:rateless_encode} \\
	\decide_m &: \mc{Y}^{m\chunk} \times \posscsi{\chunk}^m \times [K] \to \{0,1\} 
		\label{eq:rateless_decide} \\
	\rdec_m &: \mc{Y}^{m\chunk} \times \posscsi{\chunk}^m \times [K] \to [N]~.
		\label{eq:rateless_decode}
	\end{align}
To encode chunk $m$, the encoding function $\renc_m$ uses the message in $[N]$ and key in $[K]$
to choose a vector of $\chunk$ channel inputs.

The decision function $\decide_m$ defines a random variable, called the \textit{decoding time} $\dectime$ of the rateless code:
	\begin{align}
	\dectime = \min \left\{ m : \decide_m(\mbf{y}_{1}^{mc}, \csi{1}^{m}, k) = 1 \right\}~.
	\label{eq:dectime_def}
	\end{align}
Let $\possdec = \{\Mlo, \Mlo+1, \ldots, \Mhi\}$ \label{def:dectime_support} be the smallest interval containing the support of $\dectime$.  The set of possible (empirical) rates for the rateless code are given by $\{ (m \chunk)^{-1} \log N : m \in \possdec\}$.

We can define decoding regions for the rateless code at a decoding time $\dectime = M$.  Note that if $\dectime = M$ we have $\decide_{M}(\mbf{y}_{1}^{M \chunk}, \csi{1}^{M},k) = 1$.  For message $i$, key $k$ and side information vector $\csi{1}^{M}$ we can define a decoding region:
	\begin{align}
	D_{i,k}(\csi{1}^{M})
	= 
		\big\{ \mbf{y}_{1}^{M \chunk} : & 
			\decide_{M}(\mbf{y}_{1}^{M \chunk}, \csi{1}^{M},k) = 1, \nonumber \\
			&
			\rdec_{M}(\mbf{y}_{1}^{M \chunk}, \csi{1}^{M},k) = i
			\big\}~.
	\end{align}
The \textit{maximal} and \textit{nosy noise error} for a $(\chunk,N, K)$ rateless code at decoding time $\dectime = M$ are, respectively, 
	\begin{align}
	\stdmaxerr(M,\mbf{s},\csi{1}^{M}) & \nonumber \\
	&\hspace{-1.5cm}
	= \max_{i \in [N]} 
		\frac{1}{K} \sum_{k=1}^{K}  
		\left(
		1 - W^{M \chunk} \left(
		 D_{i,k}(\csi{1}^{M})
		 \Big| 
		 \renc_{1}^{M}(i, k), 
		 \mbf{s}_{1}^{M \chunk} 
		\right)
		\right)
	\label{eq:rateless_maxerr}
	\\
	\nosymaxerr(M,J,\csi{1}^{M}) & \nonumber \\
	&\hspace{-1.5cm}
	= \max_{i \in [N]} 
		\frac{1}{K} \sum_{k=1}^{K}  \nonumber \\
		&\hspace{-0.5cm}
		\bigg(
		1 - W^{M \chunk} \bigg(
		D_{i,k}(\csi{1}^{M}) 
		\nonumber \\
		&\hspace{2cm}
		 \Big|~ 
		 \renc_{1}^{M}(i, k), 
		 J_M(i, \renc_{1}^{M}(i, k))
		\bigg)
		\bigg)~.
	\label{eq:rateless_nosyerr}
	\end{align}
Here $J = (J_1, \ldots, J_M)$ and $J_M : [N] \times \mc{X}^{M \chunk} \to \mc{S}^{M \chunk}$ is the adversary's strategy.  Note that in these error definitions we do not take the maximum over all $\mbf{s}$ or $\mbf{J}$, because the rate and error at which we decode will depend on the realized state sequence, in contrast to the point-to-point AVC errors in (\ref{eq:maxerrdef}) and (\ref{eq:nosyerr}).

Because we consider rateless codes with finite total blocklength $n$, under some state sequences the decoder may never decide to decode.  Intuitively, this is because the channel is too noisy.  In order to quantify the performance of a rateless code, we must specify the set of state sequences for which the code will decode.

\subsection{Partial channel state information \label{sec:csi}}

Suppose that during the $m$-th chunk of channel uses $\{ (m-1) \chunk + 1, \ldots m \chunk\}$ the channel inputs were $\cvar{x}{m}$ and the state was $\cvar{s}{m}$.  Under the maximal error criterion, we define the average channel under $\mbf{s}$ during the $m$-th chunk by
	\begin{align*}
	V_m(y | x) = 
		\frac{1}{\chunk}
		\sum_{t = (m-1) \chunk +1}^{m \chunk} W(y | x, s_t)~.
	\end{align*}
Under the nosy noise criterion we define the average channel under $\mbf{x}$ and $\mbf{s}$ by
	\begin{align}
	V_m(y | x) 
	   = \frac{1}{N(x | \cvar{x}{m})}
	   	\sum_{t = (m-1) \chunk +1}^{m \chunk} 
			   	W(y | x_t, s_t) \mbf{1}(x_t = x)~.
	\label{eq:rob:avgchan}
	\end{align}
A receiver with full side information would learn the channel $V_m$ explicitly.  We consider instead the case where the receiver is given a set $\csi{m}$ after the $m$-th chunk, where $\csi{m}$ is a subset of channels such that $V_m(y | x) \in \csi{m}$.

We denote the set of possible values for $\csi{m}$ by $\posscsi{\chunk}$.  This is a collection of subsets of $\wstd(\scostb) \cap \mc{P}_{\chunk}(\mc{Y} | \mc{X})$ for maximal error and of $\wdep(\scostb) \cap \mc{P}_{\chunk}(\mc{Y} | \mc{X})$ for nosy noise.  We will assume a polynomial upper bound on the size of $\posscsi{\chunk}$:
	\begin{align}
	|\posscsi{\chunk}| \le \chunk^{\csiexp}~,
	\label{def:csiexp}
	\end{align}
for some $\csiexp < \infty$.

We consider two models for $\csi{m}$: in the first the decoder gets an estimate the empirical cost of the true state sequence, and in the second the decoder gets an estimate of the mutual information induced by the true channel.  For rateless codes under maximal error we will assume that the receiver gets an estimate $\normcsi{m}$ such that the true cost
	\begin{align}
	\truenorm{m} 
		= \frac{1}{\chunk} \sum_{t = (m-1) \chunk + 1}^{m \chunk} 
		l(s_t)
	\label{eq:truenorm}
	\end{align}
satisfies $\truenorm{m} \le \normcsi{m} \le \truenorm{m} + \epsilon$.  The CSI set is then
	\begin{align}
	\csi{m} &= \Bigg\{ 
		V(y | x) = \frac{1}{\chunk} \sum_{i=1}^{\chunk} W(y | x, \hat{s}_i) 
		\nonumber \\
		&\hspace{3cm} 
		: \scostf(\hat{\mbf{s}}) \le \scostf(\cvar{s}{m}) + \chunk \stdrleps
		\Bigg\}~.
	\label{eq:std:consistentCSI}
	\end{align}
We call such CSI \textit{$\epsilon$-cost-consistent}.

For rateless codes under nosy maximal error, we will say a CSI sequence is \textit{$\epsilon$-consistent} for input $P$ if
	\begin{align}
	\empMI{P}{V_m} - \min_{V \in \csi{m}} \empMI{P}{V} \le \epsilon~.
	\label{eq:consistent}
	\end{align}
Our rateless codes for nosy maximal error will assume the CSI sequence is $\epsilon$-consistent.

In our rateless code constructions we use a threshold rule on the minimum mutual information of the channel consistent with the side information $\csi{1}, \csi{2}, \ldots$.  Once the receiver decides to decode, it implements the decoding rule for the rateless code.  The decoder for the codes in Section \ref{sec:std_rateless} is a maximum mutual information (MMI) decoder, and a natural question is whether the channel outputs can be used to decide the decoding time.  One way to do this is for the decoder to restrict the side information set $\csi{m}$ to those channels consistent with the output.

\section{Main results and contributions \label{sec:main}}

\subsection{Point-to-point AVCs}

Our first main result is Theorem \ref{thm:nosynoise}, which is a characterization of the capacity of the AVC with nosy noise.  The proof is given in Section \ref{sec:hashing}.

\begin{theorem} \label{thm:nosynoise}
Let $\mc{W}$ be an AVC with state cost function $\scostf(\cdot)$ and cost constraint $\scostb$.  Then $\capdep(\scostb)$ is the randomized coding capacity of the AVC with nosy noise:
	\begin{align}
	\hat{C}_r(\scostb) = \capdep(\scostb)~.
	\end{align}
Furthermore, for any $\nosyeps > 0$, there exists an $n$ sufficiently large such that the sequence of rate-key size pairs $(R, K(n))$ is achievable with nosy maximal error $\hat{\varepsilon}_r(n)$, where $n^2 \le K(n) \le \exp( n \nosyeps)$ and 
	\begin{align}
	R &= \capdep(\scostb) - \nosyeps \\
	\nosymaxerr(n) &\le \exp(- n \hat{E}(\nosyeps)) + 
		\frac{12 n \capdep(\scostb) \log |\mc{Y}|  }{
				\nosyeps \sqrt{K(n)} \log K(n)}~,
		\label{eq:rob_pt_err}
	\end{align}
where $\hat{E}(a) > 0$ for $a > 0$. 
\end{theorem}

This theorem is proved by first constructing list-decodable codes with constant list size for cost-constrained AVCs.  These list-decodable codes can be combined with a message-authentication scheme due to Langberg \cite{Langberg:04focs} in Lemma \ref{lem:randfromlist}, which shows that the a secret key can be used to disambiguate the list.  Because $\wstd(\scostb) \subseteq \wdep(\scostb)$, in general we have $\capdep(\scostb) \le \capstd(\scostb)$.  In some cases equality can hold, as in the following example.

\begin{example}[Bit-flipping (mod-two adder) \label{ex:bitflip}]
Consider an AVC with input alphabet $\mc{X} = \{0,1\}$, state alphabet $\mc{S} = \{0,1\}$ and output alphabet $\mc{Y} = \{0,1\}$, with
	\begin{align*}
	y = x \oplus s~,
	\end{align*}
where $\oplus$ denotes addition modulo two.  This is a ``bit-flipping AVC'' in which the jammer can flip the input ($s = 1$).  We choose $l(s) = s$ so that the state constraint $\scostb < 1/2$ bounds the fraction of bits which can be flipped by the jammer.  It has been shown \cite{CsiszarN:88constraints,Langberg:04focs} that
	\begin{align*}
	\capstd(\scostb) &= 1 - \hb(\scostb) \\
	\capdep(\scostb) &= 1 - \hb(\scostb)~,
	\end{align*}
where $\hb(t) = -t \log t - (1-t) \log (1-t)$ \label{def:hb} is the binary entropy function.  In this case, we have $\capstd(\scostb) = \capdep(\scostb)$.  Furthermore, the capacity under randomized coding and maximal error $C_r(\scostb) = \capstd(\scostb)$ and the capacity under randomized coding and nosy noise is $\hat{C}_r(\scostb) = \capdep(\scostb)$.
\end{example}

Although for this bit-flipping example the two max-min expressions have the same value, this is not the case for general AVCs.  In the previous example the addition was taken over the finite field $\mathbb{F}_2$.  If we instead take the addition over the integers the two quantities are different.

\begin{example}[Real adder  \label{ex:realadd}]
Consider an AVC with input alphabet $\mc{X} = \{0,1\}$, state alphabet $\mc{S} = \{0,1\}$ and output alphabet $\mc{Y} = \{0,1,2\}$, with
	\begin{align*}
	y = x + s~.
	\end{align*}
We choose $l(s) = s$ so that the constraint $\scostb$ on the jammer bounds the weight of its input.  For this channel, if $\scostb \ge 1/2$ Csisz\'{a}r and Narayan \cite{CsiszarN:88constraints} showed that $C_r(\scostb) = 1/2$ and is achieved with $P = (1/2, 1/2)$.  However, in the case of nosy noise the capacity is lower when $\scostb > 1/2$ because the jammer can see the codeword, it can selectively set the output to be $1$ if $P = (1/2, 1/2)$.  We have \cite{Sarwate:08thesis}:
	\begin{align*}
	\capdep(\scostb) = h_b\left( \frac{1 - \scostb}{2} \right)
			- \frac{1 + \scostb}{2} h_b\left( \frac{2 \scostb}{1 + \scostb} \right)~.
	\end{align*}
Thus we can see that $\hat{C}_r(\scostb) = \capdep(\scostb) < 1/2$.
\end{example}

\subsection{Rateless coding}

Theorems \ref{thm:std_rateless} and \ref{thm:rob_rateless} provide achievable strategies for rateless coding over channels with input-independent and input-dependent state, respectively.  The proofs of these theorem are given in Section \ref{pf:std_rateless} and Section \ref{pf:rob_rateless}.  To state our results in a way that makes the tradeoff between error probability and blocklength clearer, we will assume 
	\begin{align}
	\chunk(n) &= n^{1/4} 
		\label{eq:chunkassump} \\
	\Mhi(n) &= n/\chunk(n) = n^{3/4}~.
		\label{eq:mhiassump}
	\end{align}
For maximum and minimum rates $R_{\max}$ and $R_{\min}$ the number of messages is $N(n) = \exp( n R_{\min})$ and $\Mlo = \frac{R_{\min}}{R_{\max}} n^{3/4}$.

\begin{figure}
\begin{center}
\psset{unit=0.04cm}
\begin{pspicture}(-18,-10)(180,120)

\psline[arrows=->](0,0)(0,110)
\psline[arrows=->](0,0)(150,0)

\rput(0,115){rate}
\rput(160,0){time}

\psplot[plotpoints=200,algebraic=true,%
    linestyle=dotted,linewidth=1,linecolor=black]{55}{130}{480/(x-50)}

\psline(55,-3)(55,3)
\rput(55,-8){$M_{\ast} c$}

\psline(130,-3)(130,3)
\rput(130,-8){$M^{\ast} c$}

\psline(-3,6)(3,6)
\rput(-12,6){$R_{\min}$}

\psline(-3,96)(3,96)
\rput(-12,96){$R_{\max}$}

\psline(0,50)(8,45)(16,32)(24,40)(32,42)(40,18)(48,33)(54,31)(62,22)(70,33)
\psline[linestyle=dashed](0,46)(8,41)(16,28)(24,36)(32,38)(40,14)(48,29)(54,27)(62,18)(70,29)
\pscircle*(70,29){1.5}

\psframe[linewidth=0.5](100,55)(180,90)
\rput[bl](120,80){\scriptsize{decoding threshold}}
\rput[bl](120,70){\scriptsize{AVC with same cost}}
\rput[bl](120,60){\scriptsize{estimated channel}}
\psline[linestyle=dotted](105,82.5)(115,82.5)
\psline[linestyle=solid](105,72.5)(115,72.5)
\psline[linestyle=dashed](105,62.5)(115,62.5)

\end{pspicture}
\caption{Decoding rate versus time in a rateless code.  The empirical mutual information corresponding to the AVC with the true cost (solid line) varies, and the $\delta$-consistent channel estimates (dashed line) can track it.  Once the channel estimates cross the decoding threshold (dotted line), the receiver terminates transmission and tries to decode. \label{fig:channel_est}}
\end{center}
\end{figure}
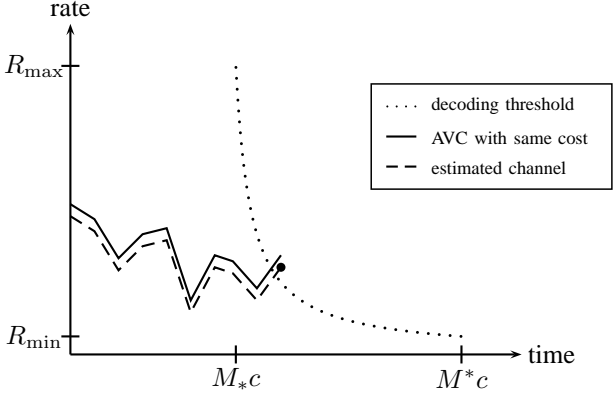

\begin{theorem}  \label{thm:std_rateless}
Let $\mc{W}$ be an AVC with state cost function $\scostf(\cdot)$.  Fix $\stdrleps > 0$, $\minrate > 0$, and input type $P \in \mc{P}(\mc{X})$ with $\min_x P(x) > 0$.  Then there is an $n_0$ sufficiently large such that for all $n > n_0$ there exists a $(\chunk(n), \exp(n \minrate), K(n))$ randomized rateless code with $K(n)/n \to \infty$ whose decoding time satisfies
	\begin{align}
	\dectime &= \min_{\Mlo \le M \le \Mhi} \bigg\{ \frac{n \minrate}{M \chunk} 
			< 
			\empMI{P}{\wstd\left( \frac{1}{M\chunk} l(\mbf{s}_1^{M\chunk}) \right)}
			\nonumber \\
		&\hspace{6cm}
			- g(\stdrleps)
			\bigg\}~,
	\label{eq:stdrl:dectime}
	\end{align}
where $g(\stdrleps) \to 0$ as $\stdrleps \to 0$.  The maximal error of the code at this decoding time satisfies
	\begin{align*}
	\stdmaxerr(\mbf{s},\csi{1}^{\dectime}) &= O\left( \frac{ n }{K(n)} \right)
	\end{align*}
for state sequences $\mbf{s}$ and $\stdrleps$-cost-consistent CSI $\csi{1}^{\dectime}$.
\end{theorem}

This theorem says that if the CSI estimates the state cost in each chunk to within $\stdrleps$, then the decoder will terminate transmission as soon as the mutual information of the channel exceeds the empirical rate $\frac{n \minrate}{M \chunk}$.  This is illustrated in Figure \ref{fig:channel_est}.  The solid line represents the mutual information of the AVC corresponding to the true state cost $l(\mbf{s}_1^{M\chunk})$, whichs is the worst-case over all state sequences whose cost is less that or equal to the true cost.   The dashed line represents the mutual information of corresponding to the estimated cost.  The dotted line is the empirical rate, so once the estimate crosses the threshold then the decoder will decode.  Furthermore, the error decays as $O(n/K(n))$.  The codebook is constructed by taking a fully randomized constant composition code that is good for an AVC, manipulating it into a rateless code with the desired properties, and reducing the common randomness using Lemma \ref{lem:elim}.

\begin{theorem}  \label{thm:rob_rateless}
Let $\mc{W}$ be an AVC with state cost function $\scostf(\cdot)$.  Fix $\minrate > 0$, $\robrlloss > 0$, input type $P \in \mc{P}(\mc{X})$ with $\min_x P(x) > 0$.  Then there is an $n_0$ sufficiently large such that for all $n>n_0,$ there exists a $(\chunk(n), \exp(n \minrate), K(n))$ rateless code whose decoding time satisfies
	\begin{align*}
	\dectime = \min_{\Mlo \le M \le \Mhi} \left\{ M : 
		\frac{n \minrate}{M \chunk}  
		< 
		\frac{1}{M} \sum_{m=1}^{M} \empMI{P}{V_m} 
				- 2 \robrlloss 
		\right\}~,
	\end{align*}
where $V_m$ is the average channel in (\ref{eq:rob:avgchan}).  The nosy maximal error at this decoding time satisfies 
	\begin{align*}
	\robrlerr(J,\csi{1}^{\dectime}) 
		&\le 
		O\left( \frac{n}{
			\robrlloss \sqrt{K} \log K}  \right)
	\end{align*}
for $K(n) = O(\exp(c))$, state sequences $\mbf{s}$ and $\robrlloss$-consistent side information given by (\ref{eq:consistent}).  
\end{theorem}

The theorem says that there exists a rateless code which can be decoded as soon as the empirical mutual information $\chunk \sum_{m=1}^{M} \empMI{P}{V_m}$ is enough to sustain the $n \minrate$ bits for the message, assuming the side information is $\robrlloss$-consistent.  This threshold is sufficient to guarantee decoding error probability that decays like $1/\sqrt{K} \log K$ for an AVC with nosy noise.

In this code, the decoder decodes each chunk of $\chunk$ channel into a list of possible messages.   As more chunks are received, the list size shrinks and the decoding time $\dectime$ is chosen to guarantee that the list size is bounded by a constant.  Lemma \ref{lem:randfromlist} shows that this code can be used as part of a randomized code in which the secret key disambiguates the list at the decoder.

\begin{example}[Bit-flipping (mod-two adder)]
Consider the mod-two additive AVC described in Example \ref{ex:bitflip} on page \pageref{ex:bitflip} where the partial side information $\csi{m}$ as an estimate $\normcsi{m}$ of the empirical Hamming weight of the state sequence $\cvar{s}{m}$.  
The receiver tracks the empirical weight of the state sequence to compute an estimate $\csitot{M}$ of the crossover probability.  Theorems \ref{thm:std_rateless} and \ref{thm:rob_rateless} both give rateless codes that can decode as soon as the estimated empirical mutual information $M \chunk (1 - \hb(\csitot{M})) $ exceeds the size of the message ($\log N$ bits).  As $R_{\min}$ can be as small as we like, these codes can work for empirical state sequences with Hamming weight arbitrarily close to $1/2$.  The realized rate is within $\epsilon$ of $1 - \hb(\csitot{M})$, but the two codes differ greatly in the dependence of the error probability on the amount of common randomness.  When the bit-flips cannot depend on the transmitted codeword, the error decays with $K^{-1}$, and when they can it decays with $(\sqrt{K} \log K)^{-1}$.  
\end{example}

\paragraph{Remarks on the example}  For the bit-flipping example, the rates guaranteed by both theorems are close to the \textit{capacity} of the AVC with the corresponding cost constraint.  However, in general this may not be the case.  Both coding schemes use a fixed input type $P$, which is is a common feature of rateless coding strategies \cite{Shulman:03commonbc,DraperFK:05rateless,EswaranSSG:07indseq} but may result in some loss in rate \cite{ShulmanF:04prior} with respect to an input distribution chosen with knowledge of the empirical state distribution.  It may be possible to adapt the channel input distribution, perhaps using ideas from universal prediction \cite{MerhavF:98universal} but we leave that for future work.  

This scheme can also be used with more general settings for the parameters of the scheme, such as the chunk size.  Finally, we can also consider the case where the side information is merely consistent.  In this setting it is hard to quantify how close the rate at which we decode will be to the true channel, since there are no guarantees on the tightness of the channel estimates.

\section{Two partial derandomization techniques \label{sec:derand}}

In this paper we are interested in the tradeoffs between error probability and the amount of common randomness available to the encoder and decoder.  In this section we will show how to exploit existing techniques partially derandomize the rateless code constructions in Theorem \ref{thm:std_rateless} and \ref{thm:rob_rateless}.  

The ``elimination technique'' is due to Ahlswede  \cite{Ahlswede:78elimination} and uses a key size of $O(n)$ bits to achieve exponential decay in the probability of error\cite{Ericson:85exponent}.  The amount of shared common randomness is on the same order as the data to be transmitted, reminiscent of Shannon's ``one-time pad'' \cite{Shannon:49secret} for cryptography.  Lemma \ref{lem:elim} applies this technique to the randomized codes of Hughes and Thomas \cite{HughesT:96exponent,ThomasH:91exponent} and quantifies an achievable tradeoff between randomization and error decay.  This may be useful in engineering applications in which sharing $O(n)$ bits of key to send $O(n)$ bits of data is unreasonable.  We will use the result to bound the common randomness needed for the rateless codes considered in Section \ref{sec:std_rateless}.

A second derandomization procedure was suggested by Langberg \cite{Langberg:04focs} for what he called an ``adversarial channel'' (in the terminology of this paper, a binary bit-flipping AVC with nosy noise).  The construction starts with a list-decodable code and creates large overlapping subsets of codewords for each key.  These sub-codebooks should be large so that the number of messages is close to the rate of the list-decodable code, and the overlap should be large so that the jammer does not learn the key from seeing the codeword.  The encoder chooses the codeword corresponding to message $m$ in the sub-codebook given by key $k$.  The decoder first uses the list-decoder to find a list of $L$ candidate codewords.  By exploiting a combinatorial construction due to Erd\"{o}s, Frankl, and F\"{u}redi \cite{ErdosFF:85cover}, the sub-codebook structure can be chosen so that with high probability, only one of the codewords in the list at the decoder is in the sub-codebook corresponding to $k$.

\subsection{Derandomization for AVCs with maximal error : ``elimination''}

\begin{lemma}[Elimination technique \cite{Ahlswede:78elimination}]   \label{lem:elim}
Let $J$ be a positive integer and let $\mbf{C}$ be an $(n,N,J)$ randomized code with $N = \exp(nR)$ whose expected maximal error satisfies 
	\begin{align*}
	\max_{\mbf{s} \in \mc{S}^n(\scostb)}
		\max_{i}
		\expe_{\mbf{C}}[ \maxerrs{ i }{ \mbf{s} }] 
	&\le \delta(n)~, 
	\end{align*}
for an AVC $\mc{W}$ with cost function $\scostf(\cdot)$ and cost constraint $\scostb$.  Then for all $\mu$ satisfying:
	\begin{align*}
	\mu \log \delta(n)^{-1} - \hb(\mu) \log 2 
		> \frac{n}{K} (R \log 2 + \log |\mc{S}|)~,
	\end{align*}
where $\hb(\mu)$ is the binary entropy function, with probability exponentially small in $n$, the $(n,N,K)$ randomized code uniformly distributed on $K$ iid copies from $\mbf{C}$ will have with maximal probability of error less than $\mu$.
\end{lemma}

The proof follows directly from the arguments in \cite{Ahlswede:78elimination} and is omitted.  In particular, if the there is a sequence of randomized codes whose errors decay exponentially:
	\begin{align*}
	\delta(n) \le \exp(-\alpha n)~, 
	\end{align*}
then a little algebra shows that we can choose the key size $K(n)$ and the error $\mu$ to satisfy
	\begin{align*}
	\mu \le \frac{n}{K(n)}~,
	\end{align*}
for some $\zeta > 0$.  In particular, the code of Hughes and Thomas \cite{HughesT:96exponent,ThomasH:91exponent} has exponentially decaying error probability, so Lemma \ref{lem:elim} shows that the randomized coding capacity $C_r(\scostb)$ is achievable with common randomness $K(n)$ polynomial in $n$, which corresponds to $O(\log n)$ bits.

\subsection{Derandomization for AVCs with nosy noise : message authentication \label{sec:hashing}}

In \cite{Sarwate:08thesis} it is shown that for any $\epsilon > 0$ and $P \in \mc{P}(\mc{X})$ with $\max_{x} P(x) > 0$, for $n$ sufficiently large there exists a list-decodable code with codewords of type $P$, rate
	\begin{align}
	R = \min_{V \in \wdep(P,\scostb)} I(P,V) - \epsilon ~,
	\label{eq:lem:constlistrate}
	\end{align}
list size
	\begin{align}
	L < \left\lfloor \frac{ 6 \log |\mc{Y}| }{ \epsilon } \right\rfloor + 1~,
	\label{eq:lem:constlistsize}
	\end{align}
and error
	\begin{align*}
	\listmaxerr &\le \exp( - n E(\epsilon) )~,
	\end{align*}
where $E(\mlconeps) > 0$.  The argument is based on those of Ahlswede\cite{Ahlswede:73list,Ahlswede:93list}.

For AVCs with nosy noise, the state can depend on the transmitted codeword.  By combining these list-decodable codes with a message authentication scheme used by Langberg \cite{Langberg:04focs}, we can construct randomized codes for this channel with limited common randomness.  The relationship between the key size, list size, and error is given by the following Lemma.

\begin{lemma}[Message Authentication \cite{Langberg:04focs}]
\label{lem:randfromlist}
Let $\mc{W}$ be an AVC and suppose we are given an $(n,N,L)$ deterministic list-decodable code and probability of error $\autherr$.  For key size $K(n)$ where $K(n)$ is a power of a prime there exists an $(n, N/\sqrt{K(n)}, K(n))$ randomized code with nosy maximal error $\nosymaxerr(\mbf{s})$ such that
	\begin{align}
	\max_{\mbf{s}} \nosymaxerr(\mbf{s}) \le \autherr + \frac{2 L \log N(n)}{\sqrt{K(n)} \log K(n)}~.
	\label{eq:hasherrbound}
\end{align}
\end{lemma}

By choosing the appropriate input distribution we can obtain our first new result : a formula for the randomized coding capacity for the AVC under nosy noise.

\begin{proof}[Proof of Theorem \ref{thm:nosynoise}]
To show the converse, note that the jammer can choose a memoryless strategy $U(s | x) \in \mc{U}(P,\scostb)$.  Choosing the worst $U$ yields a discrete memoryless channel whose capacity is $\capdep(\scostb)$, and therefore the randomized coding capacity for this channel is given by $\capdep(\scostb)$.

To show that rates below $\capdep(\scostb)$ are achievable, we first fix $K(n)$ and let $P$ be the input distribution maximizing $\capdep(\scostb)$.  We can use the previous lemma with our result on list-decodable codes to achieve the desired tradeoff.  Using (\ref{eq:lem:constlistsize}), for any $\listloss(n) > 0$ we can choose an $(n,N(n),L)$ list-decodable codebook with codewords of type $P$ such that
	\begin{align*}
	L &= \left\lfloor \frac{ 6 \log |\mc{Y}| }{ \listloss(n) } \right\rfloor + 1 \\
	N(n) &= L \exp( n (\capdep(\scostb) - \listloss(n)) )~,
	\end{align*}
and error
	\begin{align*}
	\listmaxerr &\le \exp( - n E(\listloss(n)) )~.
	\end{align*}
We can use Lemma \ref{lem:randfromlist} to construct an $(n,N(n)/\sqrt{K(n)},K(n))$ randomized code with error probability 
	\begin{align*}
	\nosymaxerr &\le \exp( - n E(\listloss(n)) ) 
		+ \frac{2 L \log N(n)}{\sqrt{K(n)} \log K(n)} \\
	&< \exp( - n E(\listloss(n)) )
		+ \frac{12 n \capdep(\scostb) \log |\mc{Y}|  }{
				\listloss(n) \sqrt{K(n)} \log K(n)}~.
	\end{align*}
The rate of this randomized code is
	\begin{align*}
	R &= \frac{1}{n} \log \frac{N(n)}{\sqrt{K(n)}} \\
	&= \capdep(\scostb) - \listloss(n) 
		- \frac{1}{n} \log \frac{\sqrt{K(n)}}{L}~.
	\end{align*} 
For any $\nosyeps >0$ and $K(n) \le \exp(n \nosyeps)$ we can choose $\listloss(n)$ small enough so that $R = \capdep(\scostb) - \nosyeps$.
\end{proof}

\subsection{An open question: converses for common randomness}

Common randomness is an important resource for coding strategies for the AVC.  The two strategies mentioned in this section show that it is sufficient to have common randomness of $O(\log n)$ bits to achieve the randomized coding capacity.  It is not clear that randomness is \textit{necessary} to achieve rates as high as the randomized coding capacity.  Because the deterministic coding capacity question is notoriously difficult, it would be of interest to prove lower bounds on the common randomness needed to achieve the randomized coding capacity.

\section{Rateless coding with cost information under maximal error \label{sec:std_rateless}}

We now prove Theorem \ref{thm:std_rateless} on rateless coding for AVCs under maximal error.  We develop a coding strategy that chooses a decoding time based on information about the cost of the actual state sequence $\mbf{s}$.  We assume the state sequence $\mbf{s}$ is fixed and estimates $\scostf(\mbf{s}_1^{M \chunk})$ are revealed to the decoder after $M \chunk$ channel uses.  The decoder picks the decoding time $\mbf{M}$ such that the empirical rate is close to the mutual information of an AVC with cost constraint $\scostf(\mbf{s}_1^{\mbf{M} \chunk})$.  We use the construction of Hughes and Thomas \cite{HughesT:96exponent} as a basis for constructing a randomized rateless code using a maximum mutual information (MMI) decoder with unbounded key size, and use Lemma \ref{lem:elim} to partially derandomize the construction.

In this section we will assume the CSI takes the form of (\ref{eq:truenorm})-(\ref{eq:std:consistentCSI}) and that it is $\epsilon$-cost-consistent.  Define
	\begin{align}
	\truetot{M} &= \frac{1}{M} \sum_{m=1}^{M} \truenorm{m} 
		\label{eq:truetot} \\
	\csitot{M} &= \frac{1}{M} \sum_{m=1}^{M} \normcsi{m}~,
		\label{eq:csitot} 
	\end{align}
be the true and estimated cost for the state sequence $\mbf{s}_1^{M\chunk}$.  The number of possible values for $\truenorm{m}$ is at most $(\chunk + 1)^{|\mc{S}|}$, which is an upper bound on the number of types on $\mc{S}$ with denominator $\chunk$.  Without loss of generality we can assume $\normcsi{m}$ takes values in the same set as $\truenorm{m}$.

\subsection{The coding strategy}

Our scheme uses a fixed maximum blocklength $n$ and we will express other parameters as functions of $n$.   For a fixed minimum rate $\minrate$, input distribution $P$, and key size $K(n)$ we will construct a randomized rateless code with chunk size $\chunk(n) = n^{1/4}$ and decoding time $\Mhi(n) = n^{3/4}$ (see  (\ref{eq:chunkassump}) and (\ref{eq:mhiassump})).  We will also use a parameter $\minrate$ which is the minimum rate of the code, so $N(n) = \exp(n \minrate)$.

\textbf{Algorithm I : Rateless coding for standard AVCs}
\begin{enumerate}
\item The encoder and decoder choose a key $k \in [K(n)]$ using common randomness.  The encoder chooses a message $i \in [N(n)]$ to transmit and maps it into a codeword $\mbf{x}(i,k) \in \mc{X}^{n}$.
\item For $m = 1, 2, \ldots, \Mlo - 1$ the encoder transmits $\cvar{x}{m}(i,k)$ in the $m$-th chunk and the decoder sets the feedback bit $\decide_{m}(\mbf{y}_{1}^{(m-1)\chunk},\normcsi{1}^{m-1},k) = 0$.
\item \label{std:enctx} For $m = \Mlo, \ldots, \Mhi = n/\chunk$, if $\decide_{m-1}(\mbf{y}_{1}^{(m-1)\chunk},\normcsi{1}^{m-1},k) = 0$, the encoder transmits $\cvar{x}{m}(i,k)$ in channel uses $(m-1) \chunk + 1, (m-1) \chunk + 2, \ldots, m \chunk$.
\item The decoder receives channel outputs $\cvar{y}{m}$ and an estimate $\normcsi{m}$ of the state cost in the $m$-th chunk.  Define the decision function $\decide_m(\mbf{y}_{1}^{m\chunk},\normcsi{1}^{m},k)$ by
	\begin{align}
		\mbf{1} \left( 
			\frac{\log N}{m \chunk} 
			< 
			\empMI{P}{\wstd(\csitot{m})}
			- \stdrlloss
		\right)~,  \label{eq:std_dec_rule}
	\end{align}
Where $\csitot{m}$ is given by (\ref{eq:csitot}).  If $\decide_m(\cdot) = 1$ then the decoder attempts to decode the received sequence, sets $\hat{i} = \rdec_{m}(\mbf{y}_{1}^{m\chunk},k)$, and feeds back a $1$ to terminate transmission.  Otherwise, the decoder feeds back a $0$ and we return to step \ref{std:enctx}) to send chunk $m+1$.
\end{enumerate}

Our code relies on the existence of a set of codewords $\{\mbf{x}(i,k)$\} which, when truncated to blocklength $m \chunk$, form a good randomized code for an AVC satisfying a given cost constraint.  The key to our construction is that the condition checked by the decision function (\ref{eq:std_dec_rule}) is sufficient to guarantee that the decoding error will be small.  In order to facilitate the analysis of the coding strategy, define the rate $R_M$ at time $M$:
	\begin{align*}
	R_M = \frac{1}{M \chunk} \log N~.
	\end{align*}

\subsection{Randomized codebook construction}

Our codebook will consist of codewords drawn uniformly from the set
	\begin{align}
	(\btyp{P}{\chunk})^{n/c} = 
		\underbrace{\btyp{P}{\chunk} 
			\times \btyp{P}{\chunk}  
			\times \cdots \btyp{P}{\chunk}}_{n/c\ \mathrm{times}}~.
		\label{eq:cc_codebook}
	\end{align}
That is, the codewords are formed by concatenating constant-composition chunks of length $\chunk$.

\begin{lemma}[Fully randomized rateless codebook]	
\label{lem:good_rateless_rand}
Let $\mc{W}$ be an AVC with cost function $\scostf(\cdot)$.  For any $\stdrlloss > 0$, $\minrate > 0$ and input distribution $P \in \mc{P}(\mc{X})$ with $\min_x P(x) > 0$, for sufficiently large blocklength $n$ there exists a randomized rateless code with $N(n) = \exp(n\minrate)$ messages whose decoding time $\dectime$ satisfies (\ref{eq:std_dec_rule}) and whose rate at $\dectime = M$ satisfies
	\begin{align}
	\frac{\log N}{M \chunk} 
			< 
			\empMI{P}{\wstd\left( \frac{1}{M \chunk} \sum_{i=1}^{M\chunk} l(s_i) \right)}
			- f(\stdrlloss)~,
		\label{eq:threshCond}
	\end{align}
for all $\mbf{s}$ and $\stdrlloss$-cost-consistent partial state information sequence $\csi{1}^{M}$, where $f(\stdrlloss) \to 0$ as $\stdrlloss \to 0$.  The error at decoding time satisfies  
	\begin{align}
	\stdmaxerr(M,\mbf{s},\csi{1}^{M}) = O( \exp(- E_3(\stdrlloss) M \chunk) )~,
	\label{eq:expurg_error}
	\end{align}
where $E_3(\stdrlloss) > 0$.
\end{lemma}

\subsection{Proof of Theorem \ref{thm:std_rateless} \label{pf:std_rateless}}

We are now ready to prove the Theorem \ref{thm:std_rateless}.  %

\begin{proof}
Fix $\stdrleps > 0$, $\minrate > 0$ and $P \in \mc{P}(\mc{X})$.  Choose $n$ sufficiently large so that the codebook-valued random variable $\mbf{C}_{\Mhi}$ that is the randomized code from Lemma \ref{lem:good_rateless_rand} satisfies (\ref{eq:expurg_error}) with $\stdrleps = \stdrlloss$ under the conditions on the state and side information in (\ref{eq:std:consistentCSI}) and (\ref{eq:stdrl:dectime}).  For each $M$, let $\mbf{C}_M$ be the the codebook truncated to blocklength $M \chunk$.

We can now draw $K(n)$ codebooks sampled uniformly from $\mbf{C}_{\Mhi}$.  Since $\mbf{C}_{\Mhi}$ truncated to blocklength $M \chunk$ is $\mbf{C}_M$, this sampling induces a sampling on $\mbf{C}_{M}$ for each $M$.  Each of these truncated codebooks has error probability exponentially small in $M \chunk$, so by Lemma \ref{lem:elim} we can choose $n$ sufficiently large and chunk size $c(n)$ so that with probability going to $1$, the error probability is at most $O(n / K(n))$ for each of the truncated codes.  Therefore a code satisfying the conditions of the Theorem exists.
\end{proof}

\subsection{An application to individual sequence channels}

One case in which the we can obtain $\delta$-cost-consistent state information is in the scheme proposed by Eswaran et al.~\cite{EswaranSSG:07indseq} for coding over a channel with individual state sequence.  The codes from this section can be used as a component in that coding scheme, which is an iterated rateless coding strategy using zero-rate feedback and unlimited common randomness.  In each iteration, the encoder and decoder use common randomness to select a rateless code and use randomized training positions to estimate the channel quality.  The rateless code uses the channel estimates to pick a decoding time.  One drawback of the scheme in \cite{EswaranSSG:07indseq} is that the amount of common randomness needed to choose the rateless code is very large.  By using the rateless code constructed in Theorem \ref{thm:std_rateless} the amount of common randomness can be reduced and can be accommodated in the zero-rate feedback link.

\section{Rateless coding for channels with input-dependent state \label{sec:rob_rateless}}
We now prove Theorem \ref{thm:rob_rateless} on rateless coding for AVCs under nosy maximal error.  The idea is to build rateless codes which are \textit{list-decodable} with constant list size at the decoding time $\dectime$.  Lemma \ref{lem:randfromlist} can be used to with these list decodable codes to construct a randomized code with small key size.

\subsection{The coding strategy}

We explicitly use information about the output sequence $\mbf{y}$ at the decoder together with the side information $\csi{m}$.  For $\robrlestgap > 0$ and distribution $P \in \mc{P}(\mc{X})$, given the $m$-th chunk of channel outputs $\cvar{y}{m}$ and the side information set $\csi{m}$, define
	\begin{align*}
	\csi{m}( \cvar{y}{m}, \epsilon) & \\
	&\hspace{-1.8cm} = 
		\left\{ V \in \csi{m} : 
			\maxvar{
				\typ{\cvar{y}{m}}
				}{
				\sum_{x} P(x) V(y|x)
				}
			< \epsilon \right\}~,
	\end{align*}
where $\maxvar{\cdot}{\cdot}$ is the total variational distance.  Although $\csi{m}( \cvar{y}{m}, \epsilon)$ depends on $P$, in our construction $P$ is fixed so we do not make this dependence explicit.

\textbf{Algorithm II : Rateless coding for ``nosy noise''}
\begin{enumerate}
\item The encoder and decoder choose a key $k \in [K]$ using common randomness.  The encoder chooses a message $i \in [N]$ to transmit and maps it into a codeword $\mbf{x}(i,k) \in \mc{X}^{n}$.
\item \label{rob:enctx} If $\tau_{m-1}(\cdot) = 0$,  the encoder transmits $\cvar{x}{m}$ in channel uses $(m-1) \chunk + 1, (m-1) \chunk + 2, \ldots, m \chunk$.
\item The decoder receives channel outputs $\cvar{y}{m}$ and the channel state information set $\csi{m}$ and
calculates the set of possible channels $\csi{m}( \cvar{y}{m}, \robrlestgap)$.  Define the decision function $\decide_m(\mbf{y}_{1}^{m\chunk},\csi{1}^{m},k)$ as
	\begin{align}
		\mbf{1} \left(
			\frac{\log N}{m \chunk} 
			< 
			\frac{1}{m} \sum_{i=1}^{m} \empMI{P}{\csi{m}( \cvar{y}{m}, \robrlestgap)} - \robrlloss
			\right)~.
	\label{eq:rob_thresh}
	\end{align}
If $\decide_m(\cdot) = 1$ then the decoder attempts to decode the received sequence, sets $\hat{i} = \rdec_{m}(\mbf{y}_{1}^{m\chunk},k)$, and feeds back a $1$ to terminate transmission.  Otherwise, it feeds back a $0$ and we return to step \ref{rob:enctx}) for chunk $m+1$.
\end{enumerate}

The rateless code developed in this section has codewords in $(\btyp{P}{\chunk})^{M}$, i.e. they have type $P$ in each chunk.  Once the decision threshold $\dectime$ is reached, the decoder list decodes the received codeword and produces a list of candidate message-key pairs.  From Lemma \ref{lem:randfromlist}, with high probability there will be only one message-key pair in the list consistent with the key used to encode the message.

\subsection{List-decodable codes}

The codebook we use is again sampled from $(\btyp{P}{\chunk})^{n/c}$ given in (\ref{eq:cc_codebook}).  In Lemma \ref{lem:rob_biglist} we show that a codebook consisting of all sequences in $\btyp{P}{\chunk}$ can be used as a list-decodable code with a list size that depends on a channel estimate at the decoder.  This list size is exponential in $\chunk$.  Therefore $(\btyp{P}{\chunk})^{M}$ can also be list-decoded using the channel estimates with list size exponential in $M \chunk$.  The decoding condition (\ref{eq:rob_thresh}) can be used to bound the list size at decoding.  The final step is to sample codewords from $(\btyp{P}{\chunk})^{n/c}$.  The subsampling ensures constant bound on the list size for all decoding times.

\begin{lemma}
\label{lem:rob_biglist}
Let $\mc{W}$ be an AVC.  For any $\robrlestgap > 0$ and $\robrlgap > 0$, $P \in \mc{P}(\mc{X})$ with $\min_x P(x) > 0$ there is a $\chunk$ sufficiently large and a function $E_1(\robrlgap)$ such that for any $\csi{} \in \posscsi{\chunk}$ the set $\btyp{P}{\chunk}$ is a list-decodable code of blocklength $\chunk$ with $N$ messages and list size $L(\csi{})$ for the AVC $\mc{W}$ under nosy maximal error with
	\begin{align}
	N &= |\btyp{P}{\chunk}| \ge \exp \left( \chunk (H(X) - \robrlgap) \right) \nonumber \\
	L(\mbf{y}_{1}^{\chunk},\csi{}) 
		&\le \exp \left( 
			\chunk \left( 
				\max_{V \in \csi{}( \mbf{y}_{1}^{\chunk}, \robrlestgap)} 
				H(X|Y) 
				+ \robrlgap 
				\right)
			\right)~, 
		\label{eq:biglist_size}
	\end{align}
and error 
	\begin{align*}
	\listmaxerr \le \exp( - \chunk \cdot E_1(\robrlgap) )~,
	\end{align*}
where $H(X)$ is calculated with respect to the distribution $P(x)$ and for $V \in \csi{}( \mbf{y}_{1}^{c}, \robrlestgap)$ the conditional entropy $H(X | Y)$ is with respect to the distribution $P(x) V(y | x)$, and $E_1(\robrlgap) > 0$ for $\robrlgap > 0$.
\end{lemma}

With the previous lemma as a basic building block, we can create nested list-decodable codes where $\chunk$ is chosen to be large enough to satisfy the conditions of Lemma \ref{lem:rob_biglist}.

\begin{lemma}[Concatenated exponential list codes]
	\label{lem:list_chunk}
Let $\mc{W}$ be an AVC.  For any $\robrlestgap > 0$ and $\robrlgap > 0$, $P \in \mc{P}(\mc{X})$ with $\min_x P(x) > 0$, there is a $\chunk$ sufficiently large such that the set $(\btyp{P}{\chunk})^{M}$ is an list-decodable code with blocklength $M \chunk$, $N_M$ messages and list size $L(\mbf{y}_{1}^{M \chunk},\csi{1}^{M})$ for $\csi{1}^{M} = (\csi{1}, \csi{2}, \ldots, \csi{M}) \in \posscsi{\chunk}^{M}$, where
	\begin{align*}
	N_M &\ge \exp \left( M \chunk (H(X) - \robrlgap) \right)
	\end{align*}
and
	\begin{align}
	L(\mbf{y}_1^c,\csi{1}^{M}) & \nonumber \\
	&\hspace{-1.5cm}
	\le \exp \left( \chunk \left( \sum_{m=1}^{M} \max_{V \in \csi{m}( \cvar{y}{m}, \robrlestgap)} H(X_m|Y_m) + M \robrlgap \right)\right)~,
	\label{eq:concatListSize}
	\end{align}
and maximal probability of error 
	\begin{align}
	\listmaxerr \le M \exp( - \chunk E_2(\robrlgap) )~,
	\label{eq:concatErr}
	\end{align}
where $H(X)$ is calculated with respect to the distribution $P(x)$ and for a channel $V \in \csi{m}( \cvar{y}{m}, \robrlestgap)$ the conditional entropy $H(X | Y)$ is with respect to the distribution $P(x) V(y | x)$, and $E_2(\robrlgap) > 0$.
\end{lemma}

Our codebook is constructed by sampling codewords from the codebook $(\btyp{P}{\chunk})^{n/c} = (\btyp{P}{\chunk})^{\Mhi}$.  Truncating this set to blocklength $M \chunk$ gives $(\btyp{P}{\chunk})^{M}$.  We want to show that for each $M$ the sampled codewords can be used in a list decodable code with constant list size $L$.  We can define for each truncation $M$, output sequence $\mbf{y}_{1}^{M \chunk}$, and side information sequence $(\mc{V}_1, \ldots, \mc{V}_M)$ a ``decoding bin''
	\begin{align*}
	B(M, \mbf{y}_{1}^{M\chunk}, \csi{1}^{M}) \subset \mc{X}^{M \chunk}~,
	\end{align*}
which is the list given by the code in Lemma \ref{lem:list_chunk}.  The size of each bin can be upper bounded by (\ref{eq:concatListSize}):
	\begin{align*}
	|B(M, \mbf{y}_{1}^{M\chunk}, \csi{1}^{M})| 
	& \\
	&\hspace{-2cm}
	\le 
	\exp\left( \chunk  \left(
		\sum_{m=1}^{M} 
			\max_{V \in \csi{m}( \cvar{y}{m}, \robrlestgap)} H(X_m|Y_m) 
		+ M \robrlgap \right)
		\right)~.
	\end{align*}

\begin{lemma}[Constant list size]
\label{lem:cc_const}
Let $\mc{W}$ be an AVC with cost function $\scostf(\cdot)$.  For any $\robrlloss > 0$, $P \in \mc{P}(\mc{X})$ with $\min_x P(x) > 0$, for sufficiently large blocklength $n$ there exists a set of $N(n) = \exp(n\minrate)$ codewords $\{\mbf{x}(j) : j \in [N] \} \subset (\btyp{P}{\chunk})^{n/c}$ such that for any CSI sequence $(\csi{1}, \csi{2}, \ldots, \csi{\Mhi})$ and channel output $\mbf{y}$ with decoding time $\dectime$ given by (\ref{eq:rob_thresh}), the truncated codebook $\{\mbf{x}_{1}^{M \chunk}(j) : j \in [N] \}$ is an list decodable code with list size $L$ satisfying
	\begin{align*}
	L \ge \frac{ 12 \log|\mc{Y}| }{ \robrlloss }~,
	\end{align*}
and maximal probability of decoding error
	\begin{align*}
	\listmaxerr(M) \le M \exp ( - \chunk E(\robrlloss) )~,
	\end{align*}
where $E(\robrlloss) > 0$.
\end{lemma}

\subsection{Proof of Theorem \ref{thm:rob_rateless} \label{pf:rob_rateless}}

\begin{proof}
We will use the codebook from Lemma \ref{lem:cc_const}.  Since the set of messages of fixed size $N$, we use the construction of Lemma \ref{lem:randfromlist}.  This makes the code, when decoded at after $\dectime = M$ chunks, an $(M\chunk,  \exp(n \minrate)/\sqrt{K(n)}, K(n))$ randomized code with probability of error
	\begin{align*}
	\robrlerr(M,\mbf{s}) \le 
		M \exp ( - \chunk E(\robrlloss) ) 
		+ \frac{2 L n \minrate}{\sqrt{K} \log K}~.
	\end{align*} 
Then we can use choose $L = 12 ( \log |\mc{Y}|)/\robrlloss$ to get
	\begin{align*}
	\robrlerr(M,\mbf{s}) \le 
		M \exp ( - \chunk E(\robrlloss) ) 
		+ \frac{24 n \minrate \log |\mc{Y}| }{
			\robrlloss \sqrt{K} \log K}~.
	\end{align*}

Finally, we must show that loss in rate is small, assuming $\robrlloss$-consistent state information.  But this follows because by (\ref{eq:consistent}), for all $m$
	\begin{align*}
	\empMI{P}{V_m} - \empMI{P}{\csi{m}} \le \robrlloss~.
	\end{align*}
Therefore the average of mutual informations in (\ref{eq:rob_thresh}) is at most $\robrlloss$ smaller than the averages with the true channels and hence we get the bound on the decoding time.
\end{proof}

\section{Discussion}

In this paper we constructed rateless codes for two different channel models with time varying state based on arbitrarily varying channels.  In the first model, the state cannot depend on the transmitted codeword, and in the second model it can.  By adapting previously proposed derandomization strategies, we showed that a sublinear amount of common randomness is sufficient.  The first approach \cite{Ahlswede:78elimination} subsamples a randomized code and the second \cite{Langberg:04focs} is based on list decoding.  The latter strategy may interesting from a practical standpoint given recent attention to list decoding with soft information \cite{KoetterV:03soft}.  The common randomness  needed for our codes can be established via a zero-rate feedback link, which means that a secure control channel of small rate is all that is required to enable reliable communication in these situations.  In particular, we can partially derandomize the construction proposed in \cite{EswaranSSG:07indseq} for communicating over channels with individual state sequences.

We also found the capacity $\hat{C}_r(\scostb)$ for AVCs with ``nosy noise.''  For these channels the jammer has knowledge of the transmitted codeword and we showed the randomized coding capacity $\hat{C}_r(\scostb)$ is equal to $\capdep(\scostb)$.  Although in some examples $\hat{C}_r(\scostb)$ may equal the capacity under maximal error $C_r(\scostb)$, in general it is smaller.  It is interesting to note that the jammer's worst strategy for nosy noise is to make a ``memoryless attack'' on the input by choosing the state $s_t$ according the the minimizing conditional distribution $U(s | x_t)$ in (\ref{eq:capdep}).  In constrast, if the jammer is given strictly causal knowledge of the input sequence, Blackwell et al.~\cite{BlackwellBT:60random} showed that the capacity is given by $\capstd$, which is also the capacity when the jammer has no knowledge of the input sequence.  Thus from the jammer's perspective, causal information about $\mbf{x}$ is as good as no knowledge, and full knowledge is as good as knowledge of the current input.

One interesting model for these point-to-point channels that we did not address is the case where the jammer has noisy access to the transmitted codeword.  This can happen, for example, when the jammer is eavesdropping on a wireless multihop channel.  Our derandomization strategies are tailored to the extreme ends of our channel model, where the jammer has no knowledge or full knowledge.  A unified coding scheme that achieves capacity for a range of assumptions on the jammer's knowledge may help unify the two approaches.  Finally, although the results in this paper are for finite alphabets, extensions to continuous alphabets and the Gaussian AVC setting \cite{HughesN:87gavc,HughesN:88vgavc,CsiszarN:91gavc} should be possible using appropriate  approximation techniques.  An interesting rateless code using lattice constructions has been proposed by Erez et al. in \cite{ErezTW:07rateless}, and it would be interesting to see if that approach can work for more robust channel models.

\section*{Acknowledgments}

We thank the anonymous reviewers for their extensive comments and suggestions which greatly clarified and simplified the manuscript.  We would also like to thank the Associate Editor, I.~Kontoyiannis, for his many efforts and constructive comments.  Finally, we also benefited from fruitful discussions with K.~Eswaran and A.~Sahai.

\appendix

Technical proofs have been deferred to this appendix.   %

\subsection{Proof of Lemma \ref{lem:good_rateless_rand}}

\begin{proof}
Fix $\stdrlloss > 0$, $\minrate > 0$ and $P$.  We will prove that for each $M \in \mc{M} = \{\Mlo, \ldots, \Mhi\}$ there exists a randomized codebook $\mbf{C}_M$ of blocklength $M \chunk$ with rate 
	\begin{align*}
	R_M = \frac{n\minrate}{M \chunk}~.
	\end{align*}
Let $\codenorm{M}$ be defined by
	\begin{align}
	R_M = \empMI{P}{\wstd(\codenorm{M})} - \stdrlloss~.
	\label{eq:LambdaM}
	\end{align}
The distribution of the codebook $\mbf{C}_M$ will be the same as the distribution of the codebook $\mbf{C}_{\Mhi}$ of blocklength $\Mhi \chunk$ truncated to blocklength $\chunk$.

\textbf{Standard randomized codebook.}  Fix $M$ and let $\mbf{A}_M$ be a randomized codebook of $A$ codewords drawn uniformly from the constant-composition set $\btyp{P}{M \chunk}$ with maximum mutual information (MMI) decoding.  Choose $A$ such that
	\begin{align*}
	\frac{1}{M \chunk} \log A < \empMI{P}{\wstd(\codenorm{M})} - \stdrlloss/2~.
	\end{align*}
From Hughes and Thomas \cite[Theorem 1]{HughesT:96exponent} the following exponential error bound holds for all messages $i$ and state sequences $\mbf{s} \in \mc{S}^{M \chunk}$ with $l(\mbf{s}) \le (M \chunk) \codenorm{M}$:
	\begin{align}
	\chunkerr{M}(\mbf{A}_M,i,\mbf{s}) & \nonumber \\
		&\hspace{-1.8cm} \le \exp \left( -M \chunk \left( 
		 E_r\left( \frac{1}{M \chunk} \log A + \stdrlloss/2, P,\codenorm{M} \right) - \stdrlloss/2 
		\right) \right) \label{eq:HTerrbnd} \\
	&\hspace{-1.8cm}
	 \stackrel{\Delta}{=} \zeta_{M}~. \nonumber 
	\end{align}
The exponent $E_r(\cdot)$ is positive as long as the first argument is smaller than $\empMI{P}{\wstd(\codenorm{M})}$.
Therefore we have the same bound on the average error:
	\begin{align*}
	\frac{1}{A} \sum_{i=1}^{A} \chunkerr{M}(\mbf{A}_M,i) 
		\le \zeta_{M}~.
	\end{align*}

\textbf{Thinning.}  Let $\mbf{B}_M$ be a random codebook formed selecting $B$ codewords from $\mbf{A}_M \cap (\btyp{P}{\chunk})^{M}$.  That is, we keep $B$ codewords which are piecewise constant-composition with composition $P$.  We declare an encoding error if $|\mbf{A}_M \cap (\btyp{P}{\chunk})^{M}| < B$.  We use a combinatorial bound from \cite{EswaranSSG:07indseq}:
	\begin{align}
		\frac{|\btyp{P}{\chunk}|^{M}}{|\btyp{P}{M \chunk}|}
		&\ge 
		\exp( - M \log(M\chunk) \eta(P) ) \nonumber \\
		&\stackrel{\Delta}{=} \gamma_M~,
		\label{eq:stdcbook:thinfactor}
	\end{align}
where $\eta(P) < \infty$ is a positive constant. Since $\mbf{A}_M$ is formed by iid draws from $\btyp{P}{M \chunk}$, the event that codeword $i$ from $\mbf{A}_M$ is also in $(\btyp{P}{\chunk})^{M}$ is distributed according to a Bernoulli random variable with parameter at least $\gamma_M$.  The size of $|\mbf{A}_M \cap (\btyp{P}{\chunk})^{M}|$ is therefore the sum of iid Bernoulli random variables and the chance of encoding error can be bounded using Sanov's Theorem \cite{CoverThomas}:
	\begin{align*}
		\prob\left( |\mbf{A}_M \cap (\btyp{P}{\chunk})^{M}| < B \right) 
		& \\
		&\hspace{-2cm}
			\le (A+1)^{2} \exp \left(- A \cdot \kldiv{B/A}{\gamma_{M}} \right)~.
	\end{align*}
Choose $B = \gamma_{\Mhi}^2 A$.  Then we can make the probability that $|\mbf{A}_M \cap (\btyp{P}{\chunk})^{M}| < B$ as small as we like and much smaller than the decoding error bound.  Furthermore, this bound holds for all $M \in \mc{M}$.  Therefore a sub-codebook of $B$ piecewise constant-composition codewords exists with high probability.

The encoder using $\mbf{B}_M$ now operates as follows : it draws a realization of a codebook and declares an error if the realization contains fewer than $B$ codewords.  If there is no encoding error it transmits the $i$-th codeword in the codebook for message $i \in [B]$.  The average error on the fraction $B/A = \gamma_{\Mhi}^2$ of preserved codewords can be at most $A/B$ times the original average error:
	\begin{align*}
	\frac{1}{B} \sum_{i=1}^{B} \chunkerr{M}(\mbf{B}_M,i) 
		&\le \frac{\zeta_{M}}{\gamma_{\Mhi}^2}~.
	\end{align*}

\textbf{Permutation.}  We now form our random codebook $\mbf{C}_M$ by taking the codebook induced by encoder using $\mbf{B}_M$ and permuting the message index.  The encoder using $\mbf{C}_M$ takes a message $i$, randomly chosen permutation $\pi$ on $[B]$, and a codebook $\mc{B}$ from $\mbf{B}_M$ and outputs the codeword $\pi(i)$ from $\mc{B}$.  The maximal error for a message $i$ in this codebook the same as average error of $\mbf{B}_M$:
	\begin{align*}
	\chunkerr{M}(\mbf{C}_M,i) &= 
		\frac{1}{B!} \sum_{\pi} \chunkerr{M}(\mbf{B}_M, \pi(i))  \\
		&= \frac{1}{B} \sum_{i=1}^{B} \chunkerr{M}(\mbf{B}_M, i)  \\
		&\le \frac{\zeta_{M}}{\gamma_{\Mhi}^2}~.
	\end{align*}
For each $M$ we can construct a randomized codebook $\mbf{C}_M$ as described above.

\textbf{Nesting.}  Now consider the codebook $\mbf{C}_{\Mhi}$ of blocklength $n = \Mhi \chunk$ and set the size of the codebook to $B$ to equal $N(n) = \exp(n \minrate)$.  We must guarantee that the errors will still be small.  Since $B = \gamma_{\Mhi}^2 A$, the rate of the codebook $\mbf{A}_{\Mhi}$ is
	\begin{align*}
	\rho_{\Mhi} &= \frac{1}{\Mhi \chunk} \log \frac{N}{\gamma_{\Mhi}^2}~.
	\end{align*}
If we truncate $\mbf{C}_{\Mhi}$ to blocklength $M \chunk$, the resulting randomized code is identically distributed to $\mbf{C}_M$.  The rate for the corresponding $\mbf{A}_{M}$ can be bounded using (\ref{eq:stdcbook:thinfactor}), (\ref{eq:chunkassump}) and (\ref{eq:mhiassump}):
	\begin{align*}
	\rho_M &= \frac{1}{M \chunk} \log \frac{N}{\gamma_{\Mhi}^2} \\
	 &\le R_M + 2 \eta(P) \frac{\Mhi}{M} \frac{\log(\Mhi \chunk)}{\chunk} \\
	 &\le R_M + 2 \eta(P) \frac{R_{\max}}{R_{\min}} \cdot \frac{\log n}{n^{1/4}}~.
	\end{align*}
Therefore we can choose $n$ sufficiently large so that the gap between $\rho_M$ and $R_M$ can be made smaller than $\stdrlloss/2$, so $\rho_M < R_M + \stdrlloss/2$.  Therefore using the definition of $\codenorm{M}$ in (\ref{eq:LambdaM}) and the fact that $\codenorm{M} \ge \csitot{M}$ we have
	\begin{align*}
	\rho_M < \empMI{P}{\wstd(\codenorm{M})} - \stdrlloss/2~,
	\end{align*}
and the exponent in (\ref{eq:HTerrbnd}) is positive.  Now, for $(\mbf{s},\{\normcsi{m}\})$ such that (\ref{eq:threshCond}) holds, the error can be bounded:
	\begin{align*}
	\stdmaxerr(M,\mbf{s},\{\normcsi{m}\}) & \\
	&\hspace{-1cm} \le \frac{\zeta_{M}}{\gamma_{\Mhi}^2} \\
	&\hspace{-1cm} \le \exp \left( -M \chunk \left( 
		 E_r\left( R_M - \stdrlloss/2, P,\codenorm{M} \right) - \stdrlloss/2 
		\right) \right) \\
		&
		 \exp\left( 2 \eta(P) \Mhi \log(\Mhi \chunk) \right) \\
	&\hspace{-1cm} = O( \exp(- E_3(\stdrlloss) M \chunk) )~.
	\end{align*}

\textbf{Rate loss.}  The last step is to compare $\empMI{P}{\wstd(\codenorm{M})}$ to the empirical mutual information induced by the true state sequence.  By assumption, the partial CSI is $\stdrlloss$-cost-consistent, so by (\ref{eq:consistent}),
	\begin{align*}
	\truetot{M} \le \csitot{M} \le \truetot{M} + \stdrlloss~.
	\end{align*}
Therefore 
	\begin{align*}
	\empMI{P}{\wstd\left(\truetot{M}\right)} 
		- \empMI{P}{\wstd\left(\csitot{M}\right)} 
	= O(\stdrlloss \log \stdrlloss^{-1})~.
	\end{align*}
By the triangle inequality and (\ref{eq:std_dec_rule}),
	\begin{align*}
	\empMI{P}{\wstd\left(\truetot{M}\right)}
		- \frac{\log N}{m \chunk}
		= O(\stdrlloss \log \stdrlloss^{-1})~.
	\end{align*}
This proves (\ref{eq:threshCond}).
\end{proof}	
	
\subsection{Proof of Lemma \ref{lem:rob_biglist}}

\begin{proof}
Fix $\robrlgap > 0$ and $\robrlestgap > 0$.  For an input distribution $P(x)$ and channel $V(y|x)$, let $P'(y)$ be the marginal distribution on $\mc{Y}$ and $V'(x|y)$ be the channel such that $P(x) V(y|x) = P'(y) V'(x|y)$.  Our decoder will output the set
	\begin{align*}
	\mc{L}(\mbf{y}_{1}^{c},\csi{}) = \bigcup_{V \in \csi{}( \mbf{y}_{1}^{c}, \robrlestgap)} \shell{V'}{(|\mc{X}| + 1) \robrlgap}{\mbf{y}}~.
	\end{align*}
The size of this set is, by a union bound, upper bounded by (\ref{eq:biglist_size}).  The list coding results in \cite{Sarwate:08thesis} show that the probability that either the transmitted codeword $\mbf{x} \notin \mc{L}(\mbf{y}_{1}^{c},\csi{})$ or
	\begin{align*}
	\mbf{y}_{1}^{c} \notin \bigcup_{V \in \csi{}( \mbf{y}_{1}^{c}, \robrlestgap)}
		\shell{V}{\robrlgap}{\mbf{x}}
	\end{align*}
is upper bounded by
	\begin{align*}
	\listmaxerr(\csi{}) \le \exp( - \chunk \cdot E_L(\robrlgap) )~
	\end{align*}
for some positive function $E_L(\robrlgap)$.

For $\chunk$ sufficiently large, the size of this list can be bounded by (\ref{eq:biglist_size}), and the error probability is still bounded by 
	\begin{align*}
	\listmaxerr(\csi{}) \le \exp( - \chunk \cdot E_L(\robrlgap) )~.
	\end{align*} 
Thus, with probability exponential in $\chunk$, this set will contain the transmitted $\mbf{x} \in \etyp{P}{\chunk}$.  Taking a union bound over the $|\posscsi{\chunk}| = \chunk^{\csiexp}$ possible values of the side information $\csi{}$ shows that
	\begin{align*}
	\listmaxerr \le \exp( - \chunk \cdot E_L(\robrlgap) - \csiexp \log \chunk)~,
	\end{align*}
which gives the exponent $E_1(\robrlgap)$.
\end{proof}
	
\subsection{Proof of Lemma \ref{lem:list_chunk}}

\begin{proof}
Choose $\chunk$ large enough to satisfy the conditions of Lemma \ref{lem:rob_biglist}.  Our decoder will operate by list-decoding each chunk separately.  Let $L_m$ be the list size guaranteed by Lemma \ref{lem:rob_biglist} for the $m$-th chunk.  Then the corresponding upper bound in $\prod_{m=1}^{M} L_m$ is the desired the upper bound on $L(\csi{1}^{M})$.  The probability of the list in each chunk not containing the corresponding transmitted chunk can be upper bounded:
	\begin{align*}
	\listmaxerr \le M \exp( -\chunk E_1(\robrlgap) )~.
	\end{align*}
As long as $\chunk$ grows faster than $\log M$ the decoding error will still decay exponentially with the chunk size $\chunk$.
\end{proof}

\subsection{Proof of Lemma \ref{lem:cc_const}}

\begin{proof}
Fix $\robrlloss > 0$.  We begin with the codebook $(\btyp{P}{\chunk})^{\Mhi}$.  The truncation of this codebook to blocklength $M \chunk$ for $M \in \mc{M}$ is the codebook defined in Lemma \ref{lem:list_chunk}.  Let $\{\mbf{Z}_j : j \in [N] \}$ be $N = \exp(n \minrate)$ random variables distributed uniformly on the set $(\btyp{P}{\chunk})^{\Mhi}$.  The decoder will operate in two steps: first it will decode the received sequence into the exponential size list $B(M, \mbf{y}_{1}^{M\chunk}, \csi{1}^{M})$ given by the decoder of Lemma \ref{lem:list_chunk}, and then it will output only those codewords in the list which match one of the sampled codewords $\{\mbf{Z}_j\}$.  Note that the decoder for Lemma \ref{lem:list_chunk} has error satisfying (\ref{eq:concatErr}).

For any $\robrlestgap > 0$ and $\robrlgap > 0$ we can choose $\chunk(n)$ sufficiently large so that for any fixed $M$, $\mbf{y}_{1}^{M\chunk}$, and $\csi{1}^{M} \in \posscsi{\chunk}^M$ that satisfy the conditions of the decoding rule in (\ref{eq:rob_thresh}) the probability that $\mbf{Z}_j$ lands in the list $B(M, \mbf{y}_{1}^{M\chunk}, \csi{1}^{M})$ output by the decoder of Lemma \ref{lem:list_chunk} is upper bounded:
	\begin{align*}
	\prob( \mbf{Z}_j \in B(M, \mbf{y}_{1}^{M\chunk}, \csi{1}^{M}) )  
		& \\
		&\hspace{-2cm}
		\le \frac{
			|B(M, \mbf{y}_{1}^{M\chunk}, \csi{1}^{M})|
			}{
			\exp\left( M \chunk \left( H(X) - \robrlgap \right) \right)
			} \\
		&\hspace{-2cm}
		\le \exp \left( - \chunk \sum_{m=1}^{M} 
			\empMI{P}{\csi{m}(\cvar{y}{m}, \robrlestgap)}
			+ 2 M \chunk \robrlgap
			\right) \\
		&\hspace{-2cm}
		\triangleq G~.
	\end{align*}	

The random variable $\mbf{1}(\mbf{Z}_j \in B(M, \mbf{y}_{1}^{M\chunk}, \csi{1}^{M}) )$ is Bernoulli with parameter smaller than $G$, so we can bound the probability that $L$ of the $N$ codewords $\{\mbf{Z}_j\}$ land in the set $B(M, \mbf{y}_{1}^{M\chunk}, \csi{1}^{M})$ using Sanov's Theorem \cite{CoverThomas}:
	\begin{align*}
	\prob\left( 
	\frac{1}{N} \sum_{j = 1}^{N} \mbf{1}(\mbf{Z}_j \in B(M, \mbf{y}_{1}^{M\chunk}, \csi{1}^{M}) ) > L/N
	\right)
	& \\
	&\hspace{-4.5cm} \le (N+1)^2 \exp \left(-N    D\left(  L/N \ \|\  G \right) \right)~.
	\end{align*}
The exponent can be written as
	\begin{align*}
   L \log\left( \frac{L/N}{G} \right) + N (1 - L/N) \log\left( \frac{1 - L/N}{1 - G} \right)~.
	\end{align*}
To deal with the $(1 - L/N) \log((1 - L/N)/(1 - G))$ term we use the inequality $- (1 - a) \log(1 - a) \le 2a$ (for small $a$) on the term $(1 - L/N) \log(1 - L/N)$
and discard the small positive term $- (1 - L/N) \log(1 - G)$:
	\begin{align}
	N D\left(  L/N \ \|\  G \right) \nonumber
	&\ge 
	L \log\left( \frac{L/N}{G} \right) - N 2(L/N) \nonumber \\
	&\hspace{-2cm} = L \log\left( \frac{L/N}{G} \right) - 2 L \nonumber \\
	&\hspace{-2cm} = L \left( - n \minrate + \chunk \sum_{m=1}^{M} 
			\empMI{P}{\csi{m}(\cvar{y}{m}, \robrlestgap)}
			- 2 M \chunk \robrlgap
		\right) \nonumber \\
	&\hspace{-1cm}
		+ L \log L - 2L~. \label{eq:rob:listsubexp}
	\end{align}
From the rule in (\ref{eq:rob_thresh}), we know that $(M, \mbf{y}_{1}^{M\chunk}, \csi{1}^{M})$ satisfies:
	\begin{align}
	n \minrate < \chunk \sum_{m=1}^{M} 
			\empMI{P}{\csi{m}(\cvar{y}{m}, \robrlestgap)} 
			- M \chunk \robrlloss~.
	\end{align}
Substituting this into (\ref{eq:rob:listsubexp}) we see that
	\begin{align*}
	N D\left(  L/N \ \|\  G \right)
	&> L \left( M \chunk \robrlloss - 2 M \chunk \robrlgap \right) 
		+ L \log L - 2 L~.
	\end{align*} 
For large enough $n$ we have the bound $(N+1)^2 \le 2 n \rho + L$.  For large enough $L$, $L \log L > 3L$, so we can ignore those terms as well.  This gives the following bound:
	\begin{align}
	\prob\left( 
	\frac{1}{N} \sum_{j = 1}^{N} \mbf{1}(\mbf{Z}_j \in B(M, \mbf{y}_{1}^{M\chunk}, \csi{1}^{M}) ) > L/N
	\right) & \nonumber \\
	&\hspace{-2in}
	\le \exp \left(- L M \chunk \left(\robrlloss - 2 \robrlgap \right) 
		+ 2 n \minrate 
		\right)~.
		\label{eq:rawlistbnd}
	\end{align}

The number of decoding bins $B(M, \mbf{y}_{1}^{M\chunk}, \csi{1}^{M})$ can be bounded by
	\begin{align*}
	\left| \left\{ 
		B(M, \mbf{y}_{1}^{M\chunk}, \csi{1}^{M}) :  
			\csi{1}^{M} \in \posscsi{\chunk}^M,\ 
			\mbf{y}_{1}^{M\chunk} \in \mc{Y}^{M \chunk}
		\right\} \right|  & \\
	&\hspace{-3cm}
		\le |\mc{Y}|^{M \chunk} \chunk^{M \csiexp}~.
	\end{align*}
Therefore we can take a union bound over all the decoding bins in (\ref{eq:rawlistbnd}) to get an upper bound of 
	\begin{align*}
	\exp \left(- L M \chunk \left(\robrlloss - 2 \robrlgap \right) 
		+ M \chunk \log |\mc{Y}| + M \csiexp \log \chunk 
		+ 2 n \minrate
		\right)~.
	\end{align*}

Since $\frac{n \minrate}{M \chunk} \le \log |\mc{Y}|$ for all $M \ge \Mlo$, we can choose $n$ and $\chunk$ sufficiently large such that the upper bound becomes
	\begin{align*}
	\exp \left(- L M \chunk \left(\robrlloss - 2 \robrlgap \right) 
		+ 4 M \chunk \log |\mc{Y}|  
		\right)~.
	\end{align*}
If $\robrlloss > 2 \robrlgap$ then we can choose
	\begin{align*}
	L > \frac{ 4 \log |\mc{Y}| }{ \left(\robrlloss - 2 \robrlgap \right) }
	\end{align*}
to guarantee that subsampling will yield a good list-decodable code for all $M \in \{\Mlo, \ldots, \Mhi\}$.  Choosing $\robrlgap = \robrlloss/3$ and $E(\robrlloss) = E_2(\robrlloss/3)$, where $E_2(\cdot)$ is from (\ref{eq:concatErr}), yields the result.
\end{proof}

\bibliographystyle{IEEEtran}
\bibliography{nosy}

\end{document}